\documentclass[a4paper,onecolumn,11pt,accepted=2025-10-26]{quantumarticle}
\pdfoutput=1
\usepackage[utf8]{inputenc}
\usepackage[english]{babel}
\usepackage[T1]{fontenc}

\usepackage{amsmath,amssymb,amsthm,amsfonts,latexsym,bbm,xspace,graphicx,float,mathtools,braket,dirtytalk}
\usepackage{enumitem}

\usepackage{yhmath}
\usepackage{subcaption}
\usepackage{sidecap}
\sidecaptionvpos{figure}{c}
\usepackage[linesnumbered,ruled,vlined]{algorithm2e}
\SetKwInput{KwInput}{Input}                
\SetKwInput{KwOutput}{Output}
\SetKwFor{RepTimes}{repeat}{times}{end}

\usepackage{arydshln}
\usepackage{nicematrix}

\usepackage[colorinlistoftodos]{todonotes}

\makeatletter
\renewcommand*\env@matrix[1][*\c@MaxMatrixCols c]{%
  \hskip -\arraycolsep
  \let\@ifnextchar\new@ifnextchar
  \array{#1}}
\makeatother

\newcommand{\trc}{\tr}

\usepackage[pagebackref,colorlinks,citecolor=blue,,linkcolor=magenta,bookmarks=true]{hyperref}
\usepackage[nameinlink,capitalize]{cleveref}
\Crefname{enumi}{Condition}{Conditions}
\crefname{enumi}{Condition}{conditions}
\crefname{algocf}{Algorithm}{Algorithms}
\Crefname{algocf}{Algorithm}{Algorithms}
\crefname{algocfline}{Line}{Lines}
\Crefname{algocfline}{Line}{Lines}

\renewcommand{\backref}[1]{}

\renewcommand{\backrefalt}[4]{%
\ifcase #1 %
\or
[p.\ #2]%
\else
[pp.\ #2]%
\fi}

\newtheorem{theorem}{Theorem}[section]
\newtheorem*{namedtheorem}{\theoremname}
\newcommand{\theoremname}{testing}

\newtheorem{lemma}[theorem]{Lemma}

\newtheorem{proposition}[theorem]{Proposition}
\newtheorem{fact}[theorem]{Fact}
\newtheorem{corollary}[theorem]{Corollary}

\newtheorem{question}[theorem]{Question}

\theoremstyle{definition}
\newtheorem{definition}[theorem]{Definition}

\newtheorem{remark}[theorem]{Remark}
\newtheorem*{remark*}{Remark}

\renewcommand{\Pr}{\mathop{\bf Pr\/}}
\newcommand{\E}{\mathop{\bf E\/}}
\newcommand{\Ex}{\mathop{\bf E\/}}

\newcommand{\tr}{\mathrm{tr}}  
\newcommand{\poly}{\mathrm{poly}}

\newcommand{\bool}{\F_2}
\DeclareMathOperator{\argmax}{argmax}

\newcommand{\C}{\mathbb C}

\newcommand{\Z}{\mathbb Z}
\newcommand{\F}{\mathbb F}

\newcommand{\eps}{\varepsilon}

\newcommand{\calD}{\mathcal{D}}

\newcommand{\calQ}{\mathcal{Q}}

\newcommand{\calV}{\mathcal{V}}

\newcommand{\calX}{\mathcal{X}}

\newcommand{\calZ}{\mathcal{Z}}

\newcommand{\weyl}{\mathrm{Weyl}}

\newcommand{\fidelity}{F}

\newcommand{\tracedistance}[1]{d_{\mathrm{tr}}(#1)}

\newcommand{\indic}[1]{1_{#1}}   

\newcommand{\abs}[1]{\lvert #1 \rvert}

\newcommand{\norm}[1]{\lVert #1 \rVert}

\newcommand{\ketbra}[2]{\ket{#1}\!\!\bra{#2}}

\renewcommand{\hat}{\widehat}

\newcommand{\sympcomp}{\perp}
\newcommand{\Xs}{\mathcal{X}}
\newcommand{\Zs}{\mathcal{Z}}

\newcommand{\mclifford}{m_\mathsf{Clifford}}
\newcommand{\mcomp}{m_\mathsf{Comp}}

\newcommand{\ignore}[1]{}

\newcommand{\anote}[1]{}
\newcommand{\jnote}[1]{}
\newcommand{\hnote}[1]{}

\newcounter{termcounter}[equation]
\renewcommand{\thetermcounter}{\the\numexpr\value{equation}+1\relax.\roman{termcounter}}
\crefname{term}{term}{terms}
\creflabelformat{term}{#2\textup{(#1)}#3}

\makeatletter
\def\term{\@ifnextchar[\term@optarg\term@noarg}
\def\term@optarg[#1]#2{%
  \textup{#1}%
  \def\@currentlabel{#1}%
  \def\cref@currentlabel{[][2147483647][]#1}%
  \cref@label[term]{#2}}
\def\term@noarg#1{%
  \refstepcounter{termcounter}%
  \textup{\thetermcounter}%
  \cref@label[term]{#1}}
\makeatother

\usepackage{enumitem}

\title{Efficient Learning of Quantum States Prepared With Few Non-Clifford Gates}

\author{Sabee Grewal}
\affiliation{The University of Texas at Austin}
\email{sabee@cs.utexas.edu}
\homepage{https://sabeegrewal.com/}
\orcid{0000-0002-8241-560X}
\author{Vishnu Iyer}
\affiliation{The University of Texas at Austin}
\email{vishnu.iyer@utexas.edu}
\homepage{https://vishnuiyer.org/}
\orcid{0000-0001-8072-1390}
\author{William Kretschmer}
\affiliation{The University of Texas at Austin}
\affiliation{Simons Institute for the Theory of Computing, UC Berkeley}
\email{kretsch@cs.utexas.edu}
\homepage{https://wkretschmer.github.io/}
\orcid{0000-0002-7784-9817}
\author{Daniel Liang}
\affiliation{The University of Texas at Austin}
\affiliation{Rice University}
\affiliation{Portland State University}
\email{daniel.liang@ll.mit.edu}
\homepage{https://daniel-you-liang.github.io/}
\orcid{0000-0002-7418-0468}

\begin{document}

\maketitle

\begin{abstract}
We give a pair of algorithms that efficiently learn a quantum state prepared by Clifford gates and $O(\log n)$ non-Clifford gates. 
Specifically, for an $n$-qubit state $\ket{\psi}$ prepared with at most $t$ non-Clifford gates, our algorithms use $\poly(n,2^t,1/\eps)$ time and copies of $\ket{\psi}$ to learn $\ket{\psi}$ to trace distance at most $\eps$. 

The first algorithm for this task is more efficient, but requires entangled measurements across two copies of $\ket{\psi}$. The second algorithm uses only single-copy measurements at the cost of polynomial factors in runtime and sample complexity. Our algorithms more generally learn any state with sufficiently large \textit{stabilizer dimension}, where a quantum state has stabilizer dimension $k$ if it is stabilized by an abelian group of $2^k$ Pauli operators. We also develop an efficient property testing algorithm for stabilizer dimension, which may be of independent interest.
\end{abstract}

\newpage
\tableofcontents
\vfill
\thispagestyle{empty}
\newpage
\section{Introduction}

\emph{Quantum state tomography} is the task of constructing a classical description of a quantum state, given copies of the state.
This task---whose study dates back to the 1950s \cite{fano1957description}---is fundamentally important in quantum theory, and finds applications in the verification of quantum technologies and in experiments throughout physics, among other things. 
For a thorough history and motivation, we refer the reader to \cite{d2003quantum, Banaszek_2013}.

The optimal number of copies to perform quantum state tomography on a $d$-dimensional quantum \emph{mixed} state is $\Theta(d^2)$ using entangled measurements \cite{o2016efficient, haah2017sample} and $\Theta(d^3)$ using single-copy measurements \cite{kueng2017low, haah2017sample, chen2022tight}.
For a quantum \emph{pure} state, $\Theta(d)$ copies are necessary and sufficient \cite{bruss1999optimal}. Alas, since the dimension $d$ grows exponentially with the system size, the number of copies consumed by state tomography algorithms quickly becomes impractical, and, indeed, learning systems of even $10$ qubits can require millions of measurements \cite{song201710}. 

There have been several approaches to circumvent the exponential scaling of quantum state tomography, which we discuss further in \cref{subsec:related-work}. For example, one can try to recover less information about the state, or make additional assumptions about the state. While these results have drastically improved copy complexities relative to general quantum state tomography, many of them remain computationally inefficient. 

In this work, we present a pure state tomography algorithm whose copy and time complexities scale in the complexity of a circuit that prepares the state.
More specifically, we assume that the circuit is described by a gate set consisting of Clifford gates (i.e., Hadamard, phase, and CNOT gates) as well as single-qubit non-Clifford gates. Such gate sets are well-studied in quantum information because they are universal for quantum computation \cite{shi2002toffoli}, and have a number of desirable properties for quantum error correction and fault tolerance \cite{knill2004faulttolerant, bravyi2005magicstates, bartolucci2023fusion}, classical simulation of quantum circuits \cite{Bravyi2019simulationofquantum}, and efficient implementation of approximate $t$-designs \cite{haferkamp2020homeopathy}.

Our main result is a tomography algorithm that scales \emph{polynomially} in the number of qubits and \emph{exponentially} in the number of non-Clifford gates needed to prepare the state.

\begin{theorem}[Informal version of \cref{cor:bell-copy-main,cor:single-copy-main}]
\label{thm:main_intro}
Given $n, t \in \mathbb{N}$, there is an algorithm that uses $\poly(n, 2^t, 1/\eps)$ time and $\poly(n, 1/\eps)$ copies of an $n$-qubit state $\ket{\psi}$, and learns $\ket{\psi}$ to trace distance $\eps$ with high probability, promised that $\ket\psi$ can be produced by Clifford gates and at most $t$ single-qubit non-Clifford gates.
\end{theorem} 

Hence, our algorithm learns in polynomial time any quantum state that can be prepared by Clifford gates and $O(\log n )$ single-qubit non-Clifford gates. 
This family of states is much more expressive than outputs of Clifford circuits alone (i.e., stabilizer states). For example, these states can form quantum $k$-designs for any constant $k$ \cite{haferkamp2020homeopathy}, whereas stabilizer states are not a $k$-design for any $k > 3$.
These states also arise naturally in quantum many-body physics as eigenstates of stabilizer Hamiltonians with a few non-commuting Pauli terms. We refer the reader to \cite{gu2024doped} for more information regarding the contexts in which these states may appear.

Although our algorithm is no longer efficient when $t$ exceeds $\omega(\log n)$, it still remains more efficient than standard pure state tomography as long as $t$ is asymptotically smaller than $n$.
While the limitation of $O(\log n )$ non-Clifford gates may seem restrictive, \Cref{thm:main_intro} is likely tight up to polynomial factors. 
This is because a polynomial dependence on $n$ and $1/\eps$ is clearly necessary, and an exponential dependence on $t$ is necessary under a plausible cryptographic assumption. In particular, if there exist linear-time constructible pseudorandom quantum states \cite{ji-pseudorandom-states2018} with exponential security, then $t$-qubit states output by quantum circuits of size $O(t)$ cannot be learned in $2^{o(t)}$ time. Though this is a strong assumption, an algorithm refuting it would be a major breakthrough in cryptography, because it would also rule out linear-time quantum-secure (weak) pseudorandom functions with exponential security, by the reduction of Brakerski and Shmueli \cite{brakerski10.1007/978-3-030-36030-6_10}.

\subsection{Our Contributions}

We, in fact, give \textit{two} algorithms to solve the learning task \cref{thm:main_intro}. The difference between the two algorithms is that one uses entangled measurements across two copies of $\ket{\psi}$, while the other measures only single copies of $\ket{\psi}$ at a time. The algorithm based on entangled measurements (given in \cref{cor:bell-copy-main}) is faster and conceptually more elegant, but it requires the ability to perform a Bell measurement across two copies of $\ket{\psi}$. So, if the copies of $\ket \psi$ are provided in an online fashion, then the algorithm needs a coherent quantum memory to store copies between trials. 

For the sake of experimental feasibility, it is often preferable to consider quantum algorithms that operate without an external quantum memory---i.e., that can only store classical information after each copy of $\ket \psi$ is measured \cite{Huang2021information,Chen2022exponential,huang22quantum}. This is how our single-copy algorithm (given in \cref{cor:single-copy-main}) operates: it uses only unentangled measurements, at the cost of a polynomially worse runtime and sample complexity.

Our algorithms also learn a more general class of states, namely: quantum states with \emph{stabilizer dimension} at least $n - t$ (\cref{def:stabilizer-dimension}). 
    Informally, a quantum state has stabilizer dimension $n-t$ if it is stabilized by an abelian group of $2^{n-t}$ Pauli operators.\footnote{Recall that an operator $U$ stabilizes a quantum state $\ket\psi$ when $U\ket{\psi} = \ket{\psi}$.} 
    We emphasize that this class of states is \emph{much larger} than the class of states preparable with Clifford gates and $t$ single-qubit non-Clifford gates. 
    Indeed, there are many states with stabilizer dimension $n - t$ that require $2^{\Omega(t)}$ gates from \emph{any} universal gate set to produce (e.g., a state of the form $\ket{0^{n - t}} \otimes \ket{\psi}$, where $\ket{\psi}$ is a Haar-random state on $t$ qubits).

In \cref{appendix:mixed-states}, we generalize our algorithm that uses two-copy measurements to handle mixed states. 
A mixed state $\rho$ has stabilizer dimension $n-t$ if $\tr(P\rho) = 1$  for every $P$ in an abelian group of $2^{n-t}$ Pauli operators.
Surprisingly, essentially all of our results generalize to mixed states with no loss in parameters. 

Finally, one might question whether the learning algorithms presented in this work follow directly from prior results on ``near-stabilizer" states, such as those in \cite{grewal_et_al:LIPIcs.ITCS.2023.64,grewal2023improved}. 
We emphasize that this is not the case. 
Our algorithms are independent of the main technical results in \cite{grewal_et_al:LIPIcs.ITCS.2023.64}. 
While we do make use of certain results from \cite{grewal2023improved} (namely, \cref{thm:p_q_duality}, \cref{fact:p-support-sympcomp}, \cref{lem:arbitrary-gate-dimension} in this work), these are not part of the main contributions of either their work or ours.

\subsection{Main Ideas}
\paragraph*{Stabilizer dimension}
As remarked earlier, our algorithms learn states with stabilizer dimension at least $n - t$ (\cref{def:stabilizer-dimension}). 
Quantum states prepared by Clifford gates and $t/2$ non-Clifford gates fall into this class because they have stabilizer dimension at least $n- t$ (\cref{lem:arbitrary-gate-dimension}).

Our first observation is that learning $\ket{\psi}$ reduces to learning the group $G$ of $2^{n-t}$ stabilizing Pauli operators of $\ket{\psi}$. In particular, we show in \Cref{lem:clifford-mapping-algorithm} that given a set of generators for $G$, we can efficiently construct a Clifford circuit $C$ such that $C\ket{\psi} = \ket{\varphi}\ket{x}$, where $\ket{x}$ is a computational basis state on $n-t$ qubits and $\ket{\varphi}$ is a general state on $t$ qubits. This construction builds on standard techniques for manipulating stabilizer tableaux, which appeared e.g.\ in the Aaronson-Gottesman algorithm \cite{aaronson2004simulation}. In some sense, this step \say{compresses} the non-Cliffordness of the state into the first $t$ qubits.\footnote{We are not the first to utilize this type of decomposition; similar techniques of compressing non-Cliffordness have appeared in \cite{arunachalam_et_al:LIPIcs.TQC.2022.3, leone-stabilizer-nullity}.} Once we know $C$, we can easily learn $\ket{x}$ by measuring $C \ket\psi$, and can learn $\ket{\varphi}$ using a tomography algorithm on $t$ qubits, which takes $2^{O(t)}$ time \cite{bruss1999optimal}.

Unfortunately, learning $G$ exactly is difficult in general, because there exist states that have no Pauli stabilizers, but that are arbitrarily close to states with a large stabilizer group (for example, consider perturbing a stabilizer state in a random direction). One of our key contributions is to make the above argument based around \Cref{lem:clifford-mapping-algorithm} robust: we show that to learn $\ket{\psi}$, it suffices to merely approximate $G$ by a (possibly different) Pauli subgroup $\hat{G}$, in a sense of approximation that we will make precise later. Roughly speaking, we require that $\hat{G}$ be sufficiently large, and that a random element of $G$ approximately stabilizes $\ket{\psi}$ (up to sign) on average:
\[
\E_{P \sim G}\left[\abs{\braket{\psi|P|\psi}}^2 \right] \ge 1 - \eps.
\]
We show that given a group $\hat{G}$ of size $2^{n-t}$ that satisfies this property, \cref{alg:reduction} finds a Clifford circuit that approximately maps $\ket{\psi}$ to a state of the form $\ket{\varphi}\ket{x}$, which allows us to learn $\ket{\psi}$ using only single-copy measurements. To find such a $\hat{G}$, we give two algorithms: one that uses entangled measurements, and another that uses only single-copy measurements.

\paragraph*{Approximating $G$ using Bell measurements}
Our first algorithm (\cref{alg:learning_weyl_bell}) for approximating $G$ utilizes \emph{Bell difference sampling}, a measurement primitive that has seen widespread use in learning algorithms relating to the stabilizer formalism \cite{montanaro-bell-sampling,gross2021schur,grewal_et_al:LIPIcs.ITCS.2023.64,grewal2023improved}. We defer the details of Bell difference sampling to \Cref{subsec:weyl-expansion-and-bell-diff-sampling} but note that it involves measuring $\ket{\psi}^{\otimes 2}$ in the Bell basis, repeating twice, and interpreting the result as a Pauli operator. A key property of Bell difference sampling, which is not too hard to show (\Cref{fact:p-support-sympcomp}), is that the sampled Pauli operators always commute with all elements of $G$. This suggests a natural approach to try to approximate $G$: Bell difference sample repeatedly, and then take $\hat{G}$ to be the commutant of the sampled Pauli operators.

\textit{A priori}, it is not at all clear why this strategy could work, because in general $\hat{G}$ may be much larger than $G$. A key technical step in our proof amounts to showing that, after $\poly(n, 1/\eps)$ Bell difference samples, with high probability, $\ket{\psi}$ must be $\eps$-close to a state that is stabilized by $\hat{G}$. In other words, if after sufficiently many samples $\hat{G}$ is larger than $G$, then this witnesses that $\ket{\psi}$ is close to a state with stabilizer dimension $n - \hat{t}$ for some $\hat{t} < t$. So, we can use the aforementioned reduction to perform tomography on $\ket{\psi}$ using $\hat{G}$, because $\hat{G}$ must approximately stabilize $\ket{\psi}$.

As a byproduct of this step in our proof, we obtain an algorithm for \textit{property testing} stabilizer dimension, which may be of independent interest.

\begin{theorem}[Informal version of \Cref{thm:property-testing-alg}]
\label{thm:property_test_informal}
There is an algorithm that, given $k \ge 1$ and copies of $\ket{\psi}$, distinguishes the following two cases: either (1) $\ket{\psi}$ has stabilizer dimension at least $k$, or (2) $\ket{\psi}$ has fidelity at most $1 - \eps$ with all such states. The algorithm uses $O\left(n/\eps\right)$ copies of $\ket{\psi}$ and $O\left(n^3 / \eps\right)$ time.
\end{theorem}

Notably, this property testing algorithm is efficient for \textit{all} choices of the stabilizer dimension $k$, unlike our learning algorithm. \Cref{thm:property_test_informal} thus greatly generalizes a recent result of Grewal, Iyer, Kretschmer, and Liang \cite{grewal2023improved} that Haar-random states are efficiently distinguishable from states with nonzero stabilizer dimension.

\paragraph*{Approximating $G$ using single-copy measurements}
Conceptually, our algorithm for approximating $G$ via single-copy measurements provides a substantial strengthening and generalization of Aaronson and Gottesman's algorithm \cite{aaronson43identifying} for learning stabilizer states from single-copy measurements.\footnote{Aaronson and Gottesman's algorithm was presented in a short video lecture \cite{aaronson43identifying}, but was never published.
A brief sketch of the algorithm and a partial derivation of its workings is presented in Aaronson's lecture notes \cite{aaronson_qis2}. Part of our work's contribution is to place their algorithm on more rigorous footing.} We start by briefly summarizing their algorithm.
First, they select a uniformly random element $C$ from the Clifford group, sample repeatedly from the measurement distribution of $C\ket{\psi}$ in the computational basis, and then iterate this process for many independently chosen $C$. Aaronson and Gottesman observe that because $C\ket{\psi}$ is a stabilizer state, its measurement distribution is supported uniformly over some affine subspace of $\F_2^n$. For most choices of $C$, this affine subspace will be all of $\F_2^n$, but with some constant probability, it will be a proper subspace of $\F_2^n$. In the latter case, learning the affine subspace\footnote{I.e., computing from samples a basis that spans the space.} reveals at least one nontrivial generator of the stabilizer group $G$ of $\ket{\psi}$. For example, if the affine subspace is $0 \times \F_2^{n-1}$ (i.e., the first qubit is always measured to be $\ket{0}$), then this reveals that $C\ket{\psi}$ is stabilized by a Pauli-$Z$ on the first qubit, and therefore $\ket{\psi}$ is stabilized by $C^\dagger (Z \otimes I^{\otimes n-1}) C$. Aaronson and Gottesman then show that repeating this process for $O(n)$ randomly chosen $C$ suffices to learn a complete set of generators for $G$, with high probability.

We find that a similar approach still works to approximately learn $G$ when $\ket{\psi}$ is no longer a stabilizer state, but rather a state with stabilizer dimension $n - t$. Indeed, our algorithm (\cref{alg:learning_weyl_single}) shares similarities with Aaronson and Gottesman's, but the analysis is considerably more involved. One major difference is that the probability of choosing a Clifford circuit $C$ that lets us learn a generator of $G$ (i.e., for which the support of $C\ket{\psi}$ is a proper affine subspace) will be smaller, on the order of $2^{-t}$ (\cref{lem:sample-in-I-Z}), and so we need to sample more random Clifford circuits to learn a complete set of generators for the stabilizer group. Furthermore, we require a significantly more careful error analysis, because the measurement distribution of $C\ket{\psi}$ need not be supported uniformly over an affine subspace for non-stabilizer states $\ket{\psi}$, and so it may be impossible to learn the full support of the distribution from a reasonable number of samples. We circumvent this issue in a similar fashion to the Bell difference sampling-based approach, by arguing that it suffices to merely \textit{approximate} the support of the distribution to within some inverse-polynomial error.  

There is another minor difference between our approach and that of Aaronson and Gottesman: we use a modified sampling technique that we call \textit{computational difference sampling}, wherein we measure $C\ket \psi$ twice in the computational basis and output the difference (i.e., bitwise XOR) of the sampled strings. Taking this difference has the effect of canceling the shift on the affine subspace over which the measurement distribution of $C\ket{\psi}$ is supported, leaving a distribution over a (non-affine) subspace of $\F_2^n$.\footnote{Note that the Bell sampling approach also involves taking two samples and computing a difference. This is for a similar reason of canceling an unwanted affine shift \cite{montanaro-bell-sampling,gross2021schur,lai2022learning}.} 
We then interpret the sampled string $x \in \F_2^n$ as corresponding to the Pauli-$X$-operator
\[
\xi \coloneqq \bigotimes_{i=1}^n X^{x_i}.
\]
We choose this convention so that the samples from computational difference sampling commute with all of the Pauli-$Z$ stabilizers, as we prove in \cref{cor:comp-diff-sampling}.

\paragraph*{Unsigned Pauli operators}
Throughout our proofs, we will actually find it more convenient to work with \textit{unsigned} Pauli operators, meaning that we view elements of the Pauli group as operators modulo phase. In doing so, we can identify the $n$-qubit Pauli group under multiplication (modulo phase) with the vector space $\F_2^{2n}$ under addition, which simplifies many computations. 
In particular, this identification reduces many subroutines of our algorithms (e.g., computing the commutant of a Pauli subgroup) into straightforward linear-algebraic computations over $\F_2$, often by way of viewing $\F_2^{2n}$ as a symplectic vector space. 
See \cref{subsec:symp-vector-spaces} and \cref{subsec:weyl-expansion-and-bell-diff-sampling} for more details on the relation between Pauli operators and $\F_2^{2n}$. 

\paragraph*{Summary}
To recap, the steps in our learning algorithms are as follows: (1) use \cref{alg:learning_weyl_bell} or \cref{alg:learning_weyl_single} to learn a Pauli subgroup $\hat{G}$ that approximates the stabilizer group of $\ket{\psi}$, then (2) use \cref{alg:reduction} to compute from $\hat{G}$ a Clifford circuit $C$ such that $C\ket{\psi} \approx \ket{\varphi}\ket{x}$, and finally learn $\ket{\varphi}\ket{x}$. While some aspects of the analysis are technical, the algorithms themselves are quite simple and could be amenable to implementation on near-term devices. Indeed, the only quantum parts of the algorithms involve measuring pairs of qubits, applying Clifford circuits, measuring in the computational basis, and performing tomography on a $t$-qubit state. So, for example, the resource requirements of our algorithms are quite comparable to those of the classical shadows protocol \cite{HKP20-classical-shadows}.

\subsection{Related Work}\label{subsec:related-work}
There is a long line of work devoted to developing near-Clifford \emph{simulation} algorithms \cite{aaronson2004simulation, BrayviPhysRevLett.116.250501, RallPhysRevA.99.062337, Bravyi2019simulationofquantum, qassim2021upperbounds}, which classically simulate quantum circuits dominated by Clifford gates.
These algorithms scale polynomially in the number of qubits and Clifford gates and exponentially in the number of non-Clifford gates. 
The main contribution of this work is to complement these classical simulation algorithms with learning algorithms that scale comparably with respect to the number of non-Clifford gates.

There are a few other classes of quantum states for which time-efficient tomography algorithms are known. 
Among these are stabilizer states \cite{aaronson43identifying,montanaro-bell-sampling}, non-interacting fermion states \cite{aaronson2023efficient}, matrix product states \cite{cramer2010efficient}, and certain classes of phase states \cite{arunachalam2022phase}. 
As a result of our work, the class of quantum states prepared by Clifford gates and $O(\log n )$ non-Clifford gates joins this list.
We note that our results strictly generalize both Aaronson and Gottesman's algorithm \cite{aaronson43identifying} and Montanaro's algorithm \cite{montanaro-bell-sampling} for learning stabilizer states. 

Lai and Cheng \cite{lai2022learning} gave an algorithm that learns a quantum state that is prepared via Clifford gates and a few $T$-gates, where the $T$-gate is the non-Clifford unitary $T = \ketbra{0}{0} + e^{i \pi / 4}\ketbra{1}{1}$.
However, their algorithm only works if the circuit $U$ that prepares the input state meets two conditions. Firstly, $U$ must be written as $C_1 T^v C_2$, where $C_1$ and $C_2$ are Clifford circuits, and $T^v = T^{v_1} \otimes \dots \otimes T^{v_n}$ for a string $v \in \F_2^n$ of Hamming weight $O(\log n )$ (i.e., there is a single layer of $O(\log n )$ $T$-gates between two Clifford circuits). 
Secondly, the $X$-matrix of the stabilizer tableau of $C_1\ket{0^n}$ must be full rank, but only when restricted to the qubits that the $T$-gates act on. Suffice it to say, their algorithm works in a highly restricted setting.
In contrast, our algorithms work for \emph{any} quantum state prepared using Clifford gates and $O(\log n )$ arbitrary non-Clifford gates (not just the $T$-gate), and with the non-Clifford gates allowed to be placed anywhere in the circuit. Our work is, therefore, a substantial generalization of Lai and Cheng's algorithm.\footnote{We also believe there are some fixable errors in the most recent version of Lai and Cheng's algorithm. 
For example, using notation that we introduce later in \Cref{subsec:weyl-expansion-and-bell-diff-sampling}, \cite{lai2022learning} assumes that Bell difference sampling draws from the distribution that we call $p_\psi$. However, Bell difference sampling actually draws from a different distribution that we call $q_\psi$, which is the convolution of $p_\psi$ with itself.}

Besides Lai and Cheng \cite{lai2022learning}, other authors have explored the complexity of some related learning problems that involve Clifford+$T$ circuits. For example, \cite{hinsche-single-t-gate} observes that, given samples from the measurement distribution of a circuit comprised of Clifford gates and a single $T$-gate, learning this distribution can be as hard as the learning parities with noise (LPN) problem. Leone, Oliviero, Lloyd, and Hamma \cite{leone-stabilizer-nullity} give algorithms for learning the dynamics of a quantum circuit $U$ comprised of Clifford gates and few $T$-gates, given oracle access to $U$. Both of these results are formally incomparable to our algorithm because the inputs and outputs of the respective learning tasks are different from ours. Nevertheless, the learning algorithm of \cite{leone-stabilizer-nullity} utilizes some similar techniques that involve compressing non-Cliffordness.

In another direction, one can reduce the computational complexity of learning by only estimating certain properties of quantum states, instead of producing an entire description of the state. 
For example, consider the \emph{shadow tomography} problem \cite{doi:10.1137/18M120275X, aaronson2019gentle, 10.1145/3406325.3451109} where, given a list of known two-outcome observables and copies of an unknown quantum state, the goal is to estimate the expectation value of each observable with respect to the unknown state. 
Aaronson \cite{doi:10.1137/18M120275X} showed that shadow tomography requires a number of copies that scales polylogarithmically in both the number of two-outcome observables and the Hilbert space dimension, but the algorithm is not computationally efficient.  
More recently, Huang, Kueng, and Preskill \cite{HKP20-classical-shadows} introduced \textit{classical shadows}, a shadow tomography algorithm that could be amenable to near-term quantum devices. 
Just like prior work, classical shadows uses exponentially fewer copies of the input state relative to state tomography, but, in general, is not computationally efficient.\footnote{There are certain settings where the classical shadows protocol is computationally efficient; see \cite{HKP20-classical-shadows} for more detail.}

Finally, an important primitive in our work is Bell difference sampling, which was introduced (implicitly) by Montanaro \cite{montanaro-bell-sampling} and then further developed by Gross, Nezami, and Walter \cite{gross2021schur}. 
In recent years, this technique has been at the core of several results involving stabilizer states and near-stabilizer states. 
For example, Bell difference sampling has been used to learn  \cite{montanaro-bell-sampling} and property test \cite{gross2021schur, grewal2023improved} stabilizer states. Bell difference sampling has also been used to prove computational lower bounds for constructing pseudorandom quantum states \cite{grewal_et_al:LIPIcs.ITCS.2023.64,grewal2023improved}. As mentioned earlier, our property testing algorithm (\Cref{thm:property_test_informal}) subsumes the result from \cite{grewal2023improved} that quantum states with nontrivial stabilizer dimension are not computationally pseudorandom.

\paragraph{Concurrent Work}
Shortly after the first version of this manuscript was posted, two works \cite{leone2023learning,hangleiter2023bell} appeared that claim results similar to our Bell sampling-based learning algorithm.

Leone, Oliviero, and Hamma \cite{leone2023learning} give an algorithm for learning output states of Clifford+$T$ circuits, whose time complexity scales polynomially in the number of qubits and exponentially in the number of $T$-gates. 
However, our result is stronger than \cite{leone2023learning} for two reasons. 
Firstly, our algorithm allows for arbitrary non-Clifford gates (not just $T$-gates). In fact, our algorithm also efficiently learns any state that is stabilized by a sufficiently large Pauli subgroup, which is a more general class of states. By contrast, as stated, \cite{leone2023learning} relies on a specific algebraic property of the $T$-gate and does not readily extend in this fashion. Secondly, \cite{leone2023learning}'s algorithm scales as $\poly(n) \cdot \exp(t)$ where $t$ is the number of $T$ gates, whereas \cref{cor:bell-copy-main} scales as $\poly(n) + \exp(t)$, which gives a polynomial in $n$ savings for sufficiently large values of $t$.

Hangleiter and Gullans \cite{hangleiter2023bell} exhibit a variety of protocols involving Bell sampling that have applications for learning and error detection. One of their results is an algorithm for learning outputs of Clifford+$T$ circuits that is similar to our \Cref{alg:reduction,alg:learning_weyl_bell}. 
While they only discuss $T$-gates, we believe that their algorithm, like ours, can also learn arbitrary states stabilized by a large group of Pauli matrices. 
However, \cite{hangleiter2023bell} lacks a thorough time complexity analysis, and largely focuses on the sample complexity. We give a more comprehensive and careful analysis of the samples and time used by our algorithms. For example, Steps 2 and 3 of \cite{hangleiter2023bell}'s algorithm corresponds to our \Cref{lemma:compute-symplectic-complement,lem:clifford-mapping-algorithm} respectively, where we give a detailed description of the computations involved and bound the runtime explicitly. Additionally, our bounds on sample complexity are tighter (e.g., compare the bound $O\left(\frac{n\log(1/\delta)}{\eps}\right)$ in \cite[Lemma 4]{hangleiter2023bell} and $O\left(\frac{n \log(n)\log(1/\delta)}{\eps}\right)$ in \cite[Lemma 5]{hangleiter2023bell} to our bound $O\left(\frac{n + \log(1/\delta)}{\eps}\right)$ in \Cref{lem:sampling}).

A later work by Chia, Lai, and Lin \cite{Chia2023} 
gives a single-copy learning algorithm analogous to our \cref{alg:learning_weyl_single}, and appeared concurrently with the update to this manuscript to add \cref{alg:learning_weyl_single}. It uses a similar technique of learning generators of the stabilizer group of $\ket{\psi}$ by repeatedly repeatedly in random Clifford bases. One difference between the two approaches is the use of adaptive measurements in \cite{Chia2023}'s algorithm (i.e., measurements that depend on the outcomes of measurements taken in a previous step). By contrast, our \cref{alg:learning_weyl_single} is non-adaptive and quartically faster in $\eps^{-1}$ dependence.
However, it is slower by a factor of roughly $n$ in both the sample and time complexity.
Nevertheless, we show in \cref{sec:appendix} that a slight modification of \cref{alg:learning_weyl_single} with one additional round of adaptive measurements (inspired by the approach of \cite{Chia2023}) lets us remove this extra factor of $n$ and get the best of both runtime dependencies. 

\section{Preliminaries}

For a set of vectors $\{v_1, \ldots, v_m\}$, we denote their span by $\langle v_1, \ldots, v_m\rangle $. 
For a positive integer $n$, define $[n] \coloneqq \{1, \ldots, n\}$. We use $\log$ to denote the natural logarithm.
For a probability distribution $\calD$ on a set $S$, we denote drawing a sample $s \in S$ according to $\calD$ by $s \sim \calD$. 
For quantum pure states $\ket\psi, \ket\phi$, let $\tracedistance{\ket\psi,\ket\phi} = \sqrt{1 - \abs{\braket{\psi|\phi}}^2}$ denote the trace distance
and $F(\ket\psi, \ket\phi) = \abs{\braket{\psi|\phi}}^2$ denote the fidelity.
As a shorthand, we write the tensor product of pure states $\ket{\psi} \otimes \ket{\varphi}$ as $\ket{\psi}\ket{\varphi}$.
We also use the following concentration inequalities. 

\begin{fact}[Chernoff bound]\label{fact:chernoff}
Let $X_1, \ldots, X_n$ be independent identically distributed random variables taking values in $\{0,1\}$. Let $X$ denote their sum and let $\mu = \E[X]$. Then for any $\delta > 0$, 
\[
\Pr\left[X \leq (1- \delta)\mu\right] \leq e^{-\delta^2 \mu / 2}.
\]
\end{fact}

\begin{fact}[Hoeffding’s inequality]\label{fact:hoeffding}
Let $X_1,\dots, X_n$ be independent identically distributed random variables subject to $a_i \le X_i \le b_i$ for all $i$. 
Let $X = \sum_{i=1}^n X_i$ denote their sum and let $\mu = \E[X]$. 
Then for any $t \ge 0$,
\[
\Pr[X - \mu \geq t] \le \exp\left(- \frac{2t^2}{\sum_{i=1}^n (b_i - a_i)^2} \right)
\]
and
\[
\Pr[\abs{X - \mu} \ge t] \le 2 \exp\left(- \frac{2t^2}{\sum_{i=1}^n (b_i - a_i)^2} \right).
\]
\end{fact}

We require the following lemma, similar forms of which appeared in other works related to learning near-stabilizer states, including \cite[Lemma 5.9]{grewal2023improved} and \cite[Lemma 5]{hangleiter2023bell}.
\begin{lemma}
\label{lem:sampling}
Let $\calD$ be a distribution over $\F_2^d$, let $\eps, \delta \in (0,1)$, and suppose
\[
m \geq \frac{2\log(1/\delta) + 2d}{\eps}.
\]
Let $A \subseteq \F_2^d$ be the subspace spanned by $m$ independent samples drawn from $\calD$. Then
\[
\sum_{x \in A} \calD(x) \geq 1 - \eps
\]
with probability at least $1-\delta$.
\end{lemma}
\begin{proof}
For samples $x_1, \ldots, x_m \in \F_2^{d}$ drawn from $\calD$, 
let $A_i = \langle x_1, \ldots, x_{i}\rangle$ be the subspace spanned by the first $i$ samples for arbitrary $0 \leq i \leq m$, with the convention that $A_0$ is the trivial subspace.
Define the indicator random variable $X_i$ as 
\[
X_i = \begin{cases}
1 & \text{if $x_i \in \F_2^{d} \setminus A_{i-1}$ or $\sum_{x \in A_{i-1}} \calD(x) \geq 1 - \eps$.} \\
0 & \text{otherwise.} \\
\end{cases}
\]

Intuitively, $X_i = 1$ indicates that either the sample $x_i$ has increased the span of $A_i$ (i.e., $A_{i-1} \subset A_i)$ or that $A_{i-1}$ already accounts for a $1-\eps$ fraction of the mass of $\calD$.
If $A_{i-1}$ does account for a $1-\eps$ fraction of the mass of $\calD$, then $X_i = 1$ with probability $1$. 
If not, then the probability mass on $\F_2^{d}\setminus A_{i-1}$ is at least $\eps$, and therefore, $X_i = 1$ with probability at least $\eps$.
In both cases, $\Pr[X_i = 1] \geq \eps$, and therefore $\mu \coloneqq \E[\sum_i X_i] \geq m \eps$.\footnote{A careful observer will note that, since the $X_i$ are dependent on previous samples, the Chernoff bound does not necessarily hold. However, we can replace each $X_i$ with the random variable $Y_i$ such that $\{Y_i\}$ consists of i.i.d Bernoulli random variables with $\Pr[Y_i = 1] = \eps$. Since $\Pr[\sum_i Y_i < d] \geq \Pr[\sum_i X_i < d]$, we can simply apply the Chernoff bound to $\sum_i Y_i$ instead. Put another way, the dependence \emph{only helps}.}

Once $\sum_{i=1}^m X_i \geq d$, the subspace $A_m$ must account for a $1 - \eps$ fraction of the mass of $\calD$. 
This is because the dimension of $\F_2^{d}$ is $d$ and so the span can only expand $d$ times. 
Set $\gamma \coloneqq 1 - \frac{d}{\mu}$.
By a Chernoff bound (\cref{fact:chernoff}), 
\begin{align*}
\Pr\left[\sum_{i = 1}^m X_i < d\right]
&= \Pr\left[\sum_{i = 1}^m X_i < (1- \gamma) \mu\right]\\
&\leq \exp\left( -\frac{\mu}{2}\gamma^2 \right) \\
&= \exp\left( -\frac{\mu}{2} - \frac{d^2}{\mu} + d  \right) \\
&\leq \exp\left( -\frac{m\eps}{2} + d  \right).
\end{align*}
Hence, as long as 
\[
m \geq \frac{2\log(1/\delta) + 2d}{\eps},
\]
$A_m$ will account for a $1-\eps$ fraction of the $\calD$-mass with probability at least $1 - \delta$.
\end{proof}

\subsection{Stabilizer States and Clifford Circuits}

The $n$-qubit Pauli group 
is the set $\{\pm 1, \pm i\} \times \{I, X, Y, Z\}^{\otimes n}$. 
We refer to Hadamard, phase, and $\mathrm{CNOT}$ gates as Clifford gates, where $\ket{0}\!\!\bra{0} + i \ket{1}\!\!\bra{1}$ is the phase gate and $\mathrm{CNOT}$ is the controlled-not gate.
The Clifford group is the group of unitary transformations comprised only of Clifford gates, and we refer to a unitary transformation in the Clifford group as a Clifford circuit. We note that a uniformly random element of the Clifford group can be sampled efficiently, by the works of \cite{Koenig2014efficiently,Brayvi2021hadamard,Berg2021simple}.

\begin{proposition}[\cite{Berg2021simple}]\label{prop:random-clifford-sample}
    There is a classical algorithm that samples a uniformly random element of the $n$-qubit Clifford group and outputs a Clifford circuit implementation in time $O(n^2)$.
\end{proposition}

A unitary transformation $U$ \textit{stabilizes} a state $\ket{\psi}$ when $U\ket{\psi} = \ket{\psi}$.
A stabilizer state is a quantum state that is stabilized by an abelian group of $2^n$ Pauli operators, or equivalently, any quantum state reachable from $\ket{0^n}$ by applying a Clifford circuit.
For more background on the stabilizer formalism, see, e.g., \cite{Got97-thesis, nielsen2002quantum}.

Any gate set comprised of Clifford gates and just one single-qubit non-Clifford gate is universal for quantum computation \cite{shi2002toffoli}. 
We introduce the following definition for convenience, borrowing terminology from \cite{leone-stabilizer-nullity}.

\begin{definition}[$t$-doped Clifford circuits]
A $t$-doped Clifford circuit is a quantum circuit comprised only of Clifford gates (i.e., Hadamard, phase, and $\mathrm{CNOT}$) and at most $t$ single-qubit non-Clifford gates. 
\end{definition}

\subsection{Symplectic Vector Spaces}\label{subsec:symp-vector-spaces}
Throughout this work, $\F_2^{2n}$ is equipped with a canonical symplectic bilinear form that we term the symplectic product. 

\begin{definition}[Symplectic product]\label{def:symplectic-inner-product}
For $x,y \in \F_2^{2n}$, we define the \emph{symplectic product} as $[x,y] = x_1 \cdot y_{n+1} + x_2\cdot y_{n+2} + ... + x_n \cdot y_{2n} + x_{n+1} \cdot y_1 + x_{n+2} \cdot y_2 + ... + x_{2n} \cdot y_n$, where all operations are performed over $\F_2$.
\end{definition}

Similar to the orthogonal complement for the standard inner product,
the symplectic product gives rise to a \emph{symplectic complement}.

\begin{definition}[Symplectic complement]\label{def:symplectic-complement}
Let $A \subseteq \F_2^{2n}$ be a subspace. The \emph{symplectic complement} of $A$, denoted by $A^\sympcomp$, is defined by
\[
A^\sympcomp \coloneqq \{x \in \F_2^{2n} : \forall a \in A,\, [x,a] = 0 \}.
\]
\end{definition}
A subspace $A \subset \F_2^{2n}$ is \emph{isotropic} when for all $x,y \in T$, $[x,y] = 0$.
A subspace $A \subset \F_2^{2n}$ is \emph{coisotropic} when $A^\sympcomp$ is isotropic. 

To aid intuition, we present the following facts about the symplectic complement, many of which are similar to the more familiar orthogonal complement.
\begin{fact}\label{fact:sympcomp}
Let $A$ and $B$ be subspaces of $\F_2^{2n}$. Then:
\begin{itemize}
\item $A^\sympcomp$ is a subspace.
\item $(A^\sympcomp)^\sympcomp = A$.
\item $\abs{A} \cdot \abs{A^\sympcomp} = 4^n$, or equivalently $\dim A + \dim A^\sympcomp = 2n$.
\item $A \subseteq B \iff B^\sympcomp \subseteq A^\sympcomp$.
\end{itemize}
\end{fact}


Next, we introduce subspace addition and show how it behaves with respect to other operations on subspaces (namely, subspace intersection and the symplectic complement). 

\begin{definition}[Subspace addition]
    Given subspaces $A$ and $B$ over some vector space $\calV$, let $A+B \coloneqq \{a + b : a \in A \text{ and } b \in B\}$.
\end{definition}

\begin{fact}\label{fact:A-plus-B}
    Given subspaces $A \subseteq \calV$ and $B \subseteq \calV$ over some vector space $\calV$, then the following are true:
    \begin{itemize}
        \item $A+B = B + A$.
        \item $A+B$ is a subspace with dimension $\dim(A) + \dim(B) - \dim(A \cap B)$.
        \item $A \subseteq A+B$
        \item $B \subseteq A+B$
    \end{itemize}
\end{fact}

In the remainder of this section, we prove several useful identities that relate subspace addition to subspace intersection and the symplectic complement of a subspace. 

\begin{lemma}\label{lem:A-plus-B-sympcomp-identity}
    Let $A$ and $B$ be subspaces of $\F_2^{2n}$. Then
    \[(A+B)^\sympcomp = A^\sympcomp \cap B^\sympcomp.\]
\end{lemma}
\begin{proof}
$A^\sympcomp \cap B^\sympcomp \subseteq (A+B)^\sympcomp$, since if $x \in \F_2^{2n}$ has $[x, y] = 0$ for every $y \in A + B$, then $x$ must also have symplectic product zero with all subsets of $A+B$ (note that $A$ and $B$ are subsets of $A+B$ by \cref{fact:A-plus-B}).
Conversely, if $y \in \F_2^{2n}$ satisfies $[y, a] = 0$ and $[y, b] = 0$ for every $a \in A$ and $b \in B$ then, by linearity of the symplectic product, $[y, a+b] = 0$.
Thus the two sets are equal.
\end{proof}

\begin{lemma}\label{lem:A-plus-B-sympcomp-identity-2}
    Given subspaces $A$, $B$, and $C$ over $\F_2^{2n}$ such that $C \subseteq A$ then,
    \[
    A \cap (B + C) = (A \cap B) + C
    \]
\end{lemma}
\begin{proof}
Observe that $A \cap (B + C) = \{b + c \in \F_2^{2n} : b \in B \wedge c \in C \wedge b+c \in A\}$.
By assumption $C \subseteq A$, so $c \in C$ implies $c \in A$. 
Therefore, we can conclude that $b+c \in A$ if and only if $b \in A$. This is because for any $a \in A$, $a + b \in A$ if and only if $b \in A$, by the closure of subspaces under addition. 
Therefore  
$\{b + c \in \F_2^{2n} : b \in B \wedge c \in C \wedge b+c \in A\}$
can equivalently be written as $\{b + c \in \F_2^{2n} : b \in A \wedge b \in B \wedge c \in C\}$, which precisely the set $(A\cap B) + C$.
\end{proof}

\begin{corollary}\label{cor:A-plus-B-sympcomp-identity-3}
    Let $A$ and $B$ be subspaces of $\F_2^{2n}$, where $A$ is isotropic. Then
    \[
        \left(A^\sympcomp \cap (A+B)\right)^\sympcomp = A^\sympcomp \cap \left(A+B^\sympcomp\right).
    \]
\end{corollary}
\begin{proof}
We note that if $A$ is isotropic, then, by definition, $A \subseteq A^\sympcomp$.
\begin{align*}
    \left(A^\sympcomp \cap (A+B)\right)^\sympcomp
    &= A + (A+B)^\sympcomp && (\text{\cref{lem:A-plus-B-sympcomp-identity}})\\
    &= A + \left(A^\sympcomp \cap B^\sympcomp\right) && (\text{\cref{lem:A-plus-B-sympcomp-identity}})\\
    &= A^\sympcomp \cap \left(A + B^\sympcomp\right) && (\text{\cref{lem:A-plus-B-sympcomp-identity-2}}) \qedhere\\
\end{align*}
\end{proof}

\begin{corollary}\label{cor:A-plus-B-sympcomp-Lagrangian}
    Let $A$ and $B$ be subspaces of $\F_2^{2n}$, where $A$ is isotropic and $B$ is Lagrangian. Then
    \[
    A^\sympcomp \cap (A+B)
    \]
    is Lagrangian.
\end{corollary}
\begin{proof}
    By \cref{cor:A-plus-B-sympcomp-identity-3}, we see that $A^\sympcomp \cap (A+B)$ is equal to its symplectic complement if $B$ is Lagrangian as well.
\end{proof}

\subsection{The Weyl Expansion and Bell Difference Sampling}\label{subsec:weyl-expansion-and-bell-diff-sampling}
For $x = (a, b) \in \F_2^{2n}$, the \emph{Weyl operator} $W_x$ is defined as 
\[
W_x \coloneqq 
i^{a'\cdot b'}(X^{a_1} Z^{b_1}) \otimes \dots \otimes (X^{a_n} Z^{b_n}),
\]
where $a',b' \in \Z^n$ are the embeddings of $a,b$ into $\Z^n$.
Each Weyl operator is a Pauli operator, and every Pauli operator is a Weyl operator up to a phase. 
We will often freely go back and forth between Weyl operators and $\F_2^{2n}$.

We denote by $\Xs \coloneqq \F_2^n \times 0^n$ and $\Zs \coloneqq 0^n \times \F_2^n$ the Weyl operators corresponding to Pauli-$X$ and $Z$ strings, respectively.

We use the following notation for the stabilizer group of a quantum state, ignoring global phase.

\begin{definition}[Unsigned stabilizer group]\label{def:unsigned-stabilizer-group}
Let $\weyl(\ket\psi) \coloneqq \{x \in \F_2^{2n} : W_x \ket{\psi} = \pm \ket{\psi}\}$ denote the unsigned stabilizer group of $\ket{\psi}$.
\end{definition}

The symplectic product determines commutation relations between Weyl operators. 
Specifically, the Weyl operators $W_x$ and $W_y$ commute when $[x,y] = 0$, and anticommute when $[x,y] = 1$, where $x$ and $y$ are elements of $\F_2^{2n}$. 
If $T$ is a subspace of $\F_2^{2n}$, then it is isotropic if and only if $\{W_x : x \in T\}$ is a set of mutually commuting Weyl operators. 
Furthermore, it is not difficult to show that $\weyl(\ket\psi)$ is always an isotropic subspace of $\F_2^{2n}$ (see \cref{fact:M_half_commute}).

The Weyl operators form an orthogonal basis for $2^n \times 2^n$ matrices with respect to the dot product $\langle A, B \rangle = \tr{(A^\dagger B)}$. This gives rise to the so-called \emph{Weyl expansion} of a quantum state. 

\begin{definition}[Weyl expansion]
\label{def:weyl_expansion}
Let $\ket{\psi} \in \C^{2^n}$ be an $n$-qubit quantum pure state. The Weyl expansion of $\ket\psi$ is 
\[
\ketbra{\psi}{\psi} = \dfrac{1}{2^n} \sum_{x \in F_2^{2n}}\braket{\psi|W_x|\psi} W_x.
\]
\end{definition}
Using the Weyl expansion of $\ket\psi$ and the orthogonality of the Weyl operators, observe that 
\[
1 = \tr\left(\ket{\psi}\!\!\braket{\psi|\psi}\!\!\bra{\psi}\right) = \sum_{x \in \F_2^{2n}}\frac{\braket{\psi|W_x|\psi}^2}{2^n},
\]
and therefore the squared coefficients of the Weyl expansion can be normalized to form a distribution over $\F_2^{2n}$. 
We denote this distribution by $p_\psi(x) \coloneqq 2^{-n}\braket{\psi|W_x|\psi}^2$ and refer to it as the \emph{characteristic distribution}.
Gross, Nezami, and Walter \cite{gross2021schur} introduced a measurement primitive called \emph{Bell difference sampling}, which takes four copies of a quantum state as input and outputs a sample $x \in \F_2^{2n}$ according to the distribution \[q_\psi(x) \coloneqq \sum_{a \in \F_2^{2n}} p_\psi(a) p_\psi(x + a),\]
which is the convolution of $p_\psi$ with itself.

Grewal, Iyer, Kretschmer, and Liang \cite{grewal2023improved} proved additional properties of $p_\psi$ and $q_\psi$ that are needed in this work.
First, they showed that $p_\psi$ and $q_\psi$ exhibit a strong duality property with respect to the commutation relations among Weyl operators.
The proof involves symplectic Fourier analysis so the details are omitted.

\begin{proposition}[{\cite[Theorem 3.1, 3.2]{grewal2023improved}}]
\label{thm:p_q_duality}
Let $A \subseteq \F_2^{2n}$ be a subspace, and let $\ket    \psi$ be an $n$-qubit quantum pure state.  Then 
\[
\sum_{a \in A}p_\psi(a) =  \frac{\abs{A}}{2^n}\sum_{x \in A^{\sympcomp}}p_\psi(x),   
\]
and
\[
\sum_{a \in A}q_\psi(a) = \abs{A}  \sum_{x \in A^{\sympcomp}}p_\psi(x)^2.  
\]
\end{proposition}

They also prove that the support of these distributions is contained in the symplectic complement of the unsigned stabilizer group.

\begin{fact}[{\cite[Lemma 4.3, 4.4]{grewal2023improved}}]
\label{fact:p-support-sympcomp}
    The support of $q_\psi$ is contained in $\weyl(\ket \psi)^\sympcomp$. 
\end{fact}
\begin{proof}

\begin{align*}
    \sum_{x \in \weyl(\ket\psi)^\sympcomp} q_\psi(x) 
    &= \abs{\weyl(\ket\psi)^\sympcomp}\sum_{x \in \weyl(\ket\psi)} p_\psi(x)^2 && \text{(\cref{thm:p_q_duality})}\\
    &= \abs{\weyl(\ket\psi)^\sympcomp} \frac{\abs{\weyl(\ket\psi)}}{4^n} && \text{(By definition of $\weyl(\ket\psi)$)}\\
    &= 1. && (\text{\cref{fact:sympcomp}}) \qedhere
\end{align*} 
\end{proof}

We now use \cref{thm:p_q_duality} to prove that the $q_\psi$-mass on a subspace is always bounded above by the $p_\psi$-mass.

\begin{proposition}
\label{prop:q-lower-bounds-p}
Let $A \subseteq \F_2^{2n}$ be a subspace. Then
\[\sum_{a \in A} q_\psi(a) \leq \sum_{a \in A} p_\psi(a).\]
\end{proposition}
\begin{proof}
    \begin{align*}
        \sum_{a \in A} q_\psi(a) &= \abs{A} \sum_{x \in A^\perp} p_\psi(x)^2 && \text{(\cref{thm:p_q_duality})}\\
        & \leq \frac{\abs{A}}{2^n} \sum_{x \in A^\perp} p_\psi(x) && (p_\psi(x) \leq \frac{1}{2^n})\\
        &= \sum_{a \in A} p_\psi(a).  && \text{(\cref{thm:p_q_duality})} \qedhere
    \end{align*}
\end{proof}

Finally, consider the set $M = \{x \in \F_2^{2n} : 2^n p_\psi(x) > 1/2\}$. 
We show, via the Schr\"odinger uncertainty relation, that every pair of elements in this set commutes.\footnote{In 1927, Heisenberg observed a tradeoff between knowing a particle's position and momentum, which has since been generalized in several ways. This particular uncertainty relation was derived by Schr\"odinger in 1930.}

\begin{fact}[Schr\"odinger uncertainty relation \cite{schrodinger1930uncertainty, angelow2008heisenberg}]\label{fact:uncertainty}
        For a quantum state $\rho$ and observables $A$ and $B$,
        \[
            \left(\tr(A^2 \rho) - \tr(A \rho)^2\right)\left(\tr(B^2 \rho) - \tr(B \rho)^2\right) \geq \left|\frac{1}{2}\tr\left((AB + BA) \rho\right) - \tr(A \rho)\tr(B \rho)\right|^2.
        \]
    \end{fact}

    \begin{fact}\label{fact:M_half_commute}
        Let $M = \{x \in \F_2^{2n} : 2^n p_\psi(x) > \frac{1}{2}\}$. Then for all $x, y \in M$, $[x, y] = 0$.
    \end{fact}
    \begin{proof}
        Recall that Weyl operators are Hermitian.
        Let $W_x$ and $W_y$ be two Weyl operators in $M$.
        Simple calculations tell us that \[\left(\braket{\psi | W_x^2 | \psi} - \braket{\psi | W_x | \psi}^2\right)\left(\braket{\psi | W_y^2 | \psi} - \braket{\psi | W_y | \psi}^2\right) = \left(1 - 2^n p_\psi(x)\right)\left(1 - 2^n p_\psi(y)\right) < \frac{1}{4},\]
        where we use the fact that $W_x^2 = I$ for all $x \in \F_2^{2n}$.
        Now, assume that $[x, y] = 1$ for the sake of contradiction.
        Then $W_x W_y + W_y W_x = 0$, and 
        \begin{align*}
            \left|\frac{1}{2}\braket{\psi |\left(W_x W_y + W_y W_x\right) | \psi} - \braket{\psi | W_x |\psi}\braket{\psi | W_y | \psi}\right|^2
            &= \left|\braket{\psi | W_x |\psi}\braket{\psi | W_y | \psi}\right|^2\\
            &= 4^n p_\psi(x) p_\psi(y) > \frac{1}{4},
        \end{align*}
        which, by \cref{fact:uncertainty}, is a contradiction. Hence, $[x, y] = 0$.
    \end{proof}

We note that Gross, Nezami, and Walter \cite[Figure 1]{gross2021schur} previously gave a graphical proof of \cref{fact:M_half_commute}.

\subsection{Stabilizer Dimension}

We use the following stabilizer complexity measure throughout this work. 

\begin{definition}[Stabilizer dimension]\label{def:stabilizer-dimension}
Let $\ket{\psi}$ be an $n$-qubit pure state. The \emph{stabilizer dimension of $\ket{\psi}$} is the dimension of $\weyl(\ket{\psi})$ as a subspace of $\F_2^{2n}$.
\end{definition}

Put another way, the stabilizer dimension of a quantum state is $k$ when there is an abelian group of $2^k$ Pauli matrices that stabilize the state.
The stabilizer dimension of a stabilizer state is $n$, which is maximal, and, for most states, the stabilizer dimension is $0$. 
We note that the stabilizer dimension is closely related to the stabilizer state nullity \cite{beverland2020lower}, and, in particular, for $n$-qubit states, the stabilizer dimension is $n$ minus the stabilizer state nullity. 

Roughly speaking, as one applies single-qubit non-Clifford gates to a stabilizer state, the state becomes farther from all stabilizer states. The stabilizer dimension quantifies this in the following sense. 

\begin{lemma}[{\cite[Lemma 4.2]{grewal2023improved}}]
\label{lem:arbitrary-gate-dimension}
    Let $\ket\psi$ be the output state of a $t$-doped Clifford circuit. Then the stabilizer dimension of $\ket\psi$ is at least $n-2t$.
\end{lemma}
\begin{proof}[Proof sketch]
First show that a single non-Clifford gate must commute with at least $1/4$ of the Weyl operators in $\weyl(\ket{\psi})$, which decreases the stabilizer dimension by at most $2$. The proof proceeds by induction on the non-Clifford gates in the circuit, using the fact that Clifford gates preserve stabilizer dimension.

\end{proof}

\subsection{Clifford Action on \texorpdfstring{$\F_2^{2n}$}{Subspaces}}\label{subsec:clifford-action-on-F2}

The Clifford group normalizes the Pauli group by conjugation, which induces an action on $\F_2^{2n}$ through the corresponding Weyl operators. 
To simplify notation, let $C(x)$ denote the action of a Clifford circuit $C$ on $x \in \F_2^{2n}$, and define it to be the $y \in \F_2^{2n}$ such that $W_y = \pm C W_x C^\dagger$.
Similarly, for a set $S \subseteq \F_2^{2n}$, define $C(S) \coloneqq \{C(x) : x \in S\}$.
As an example, $C(\F_2^{2n}) = \F_2^{2n}$, as all Clifford circuits normalize the Pauli group up to a $\pm 1$ phase.

Two key properties of the Clifford action are linearity (i.e., $C(x) + C(y) = C(x + y)$), and preservation of the symplectic product (i.e., $[x, y] = [C(x),C(y)]$. In fact, these two properties can be taken as a \textit{definition} of the Clifford group action, which is equivalent to the action of the symplectic group $\mathrm{Sp}(2n, \F_2)$. These properties imply, for example, that the Clifford action preserves symplectic complements and inclusions among subspaces of $\F_2^{2n}$.

Below, we prove one further basic fact about the action of Clifford circuits on symplectic vector spaces over $\F_2$.

\begin{fact}\label{fact:clifford-permutes-p-mass}
Let $\ket \psi$ be an $n$-qubit quantum state, let $C$ be a Clifford circuit, and define $\ket{\phi}\coloneqq C\ket{\psi}$. 
Then
\[
    p_{\phi}(x) = p_{\psi}(C^\dagger(x))
\]
for all $x \in \F_2^{2n}$.
\end{fact}
\begin{proof}
\[
    2^n p_{\phi}(x) 
    = \braket{\psi |  C^\dagger W_x C | \psi}^2 = 2^n p_\psi(C^\dagger(x)).\qedhere
\]
\end{proof}

\section{Linear Algebra Subroutines}

Our algorithms use two linear-algebraic subroutines, which we describe below. 
First, we give an algorithm for computing the symplectic complement of a subspace.

\begin{lemma}\label{lemma:compute-symplectic-complement}
Given a set of $m$ vectors whose span is a subspace $A \subseteq \F_2^{2n}$, there is an algorithm that outputs a basis for $A^\sympcomp$ in $O(m n \cdot \min(m,n))$ time.
\end{lemma}
\begin{proof}
    The algorithm works as follows. First, construct a $m \times 2n$ matrix whose rows are the $m$ elements of $A$ given as input. 
    Then swap the left and right $m \times n$ block submatrices, and denote the resulting matrix by $M$. Observe that for a nonzero vector $v$, $Mv = 0$ only when the symplectic product between $v$ and all vectors in $A$ is $0$. Hence, $v$ is in $A^\perp$, and the nullspace of $M$ is precisely $A^\perp$. 
    Finding a basis for the nullspace of $M$ can be done via Gaussian elimination, which takes $O(m n \cdot \min(m,n))$ time.
\end{proof}

Next, we explain how to find a Clifford circuit whose action on $\F_2^{2n}$ maps an arbitrary $d$-dimensional isotropic subspace of $\F_2^{2n}$ to the subspace $0^{2n-d} \times \F_2^{d}$.
We note that while the existence of such a Clifford circuit is not difficult to show (cf.\ \cite[Lemma 5.1]{grewal2023improved}), an explicit and efficient construction requires a bit more effort.

\begin{lemma}
\label{lem:clifford-mapping-algorithm}
Given a set of $m$ vectors whose span is a $d$-dimensional isotropic subspace $A \subset \F_2^{2n}$, there exists an efficient algorithm that outputs a Clifford circuit $C$ such that $C(A) = 0^{2n-d} \times \F_2^{d}$. The algorithm runs in $O(m n \cdot \min(m,n))$ time, and the circuit size of $C$ (i.e., the number of gates) is $O(nd)$. 
\end{lemma}
\begin{proof}
We will explain the algorithm and then prove its correctness. 
To begin, run Gaussian elimination on the set of $m$ vectors to get a basis for $H$ such that, when written as a $d \times 2n$ matrix $M = (m_{i,j})$, the matrix $M$ is in row echelon form.
This process takes $O(mn\cdot\min(m, n))$ time. 
The subspace spanned by the rows of $M$ is precisely the subspace $A$. 

The matrix $M$ is essentially a \textit{stabilizer tableau}, and therefore Clifford gates have the following effect on $M$ (for additional detail see, e.g., \cite{aaronson2004simulation}):
\begin{itemize}
\item Applying the Hadamard gate on the $i$th qubit corresponds to swapping the $i$th and $(n+i)$th columns of $M$.
\item Applying the phase gate on the $i$th qubit corresponds to adding the $i$th column of $M$ to the $(n+i)$th column of $M$.
\item Applying the $\mathrm{CNOT}$ gate with control qubit $i$ and target qubit $j$ corresponds to adding the $i$th column of $M$ to the $j$th column $M$ and adding the $(n+j)$th column of $M$ to the $(n+i)$th column of $M$.
\end{itemize}
Additionally, row operations do not change the subspace spanned by the rows of $M$ and therefore can be done freely. 

Our job now is to find a sequence of Hadamard, phase, and $\mathrm{CNOT}$ gates that maps $M$ to a matrix whose rows span the subspace $0^{2n-d} \times \F_2^{d}$; in particular,  
a matrix with the following form
\begin{align}\label{eq:matrix-target}
\begin{pmatrix}[c|cc]
   0 & 0 & I
\end{pmatrix},
\end{align}
where the first $0$ is a $d \times n$ matrix of all $0$'s, the second is a $d \times (n-d)$ matrix of all $0$'s, and the last is a $d \times d$ identity matrix.

The remainder of the algorithm works as follows. 

\begin{enumerate}
    \item 
For each row $i \in [d]$ of $M$:
\begin{enumerate}
\item For each $j \in \{i, \dots, n\}$, apply phase and Hadamard gates so that either $m_{i,j} = 1$ and $m_{i,n+j} =0$ or both are $0$.
\item  If $m_{i,i} = 0$, then find a $k \in \{i+1, \dots, n\}$ for which $m_{i,k} = 1$.\footnote{At least one such $k$ must exist, for if it didn't, then the $i$th row would be all 0's, which is impossible since the rows of $M$ are linearly independent.} Apply a $\mathrm{CNOT}$ with control qubit $i$ and target qubit $k$ so that $m_{i,i} = 1$.
\item For each $j \in \{i+1,\ldots, n\}$, if $m_{i,j} = 1$, apply $\mathrm{CNOT}$ with control qubit $i$ and target qubit $j$.
\item For $j \in \{i+1, \dots, d\}$, set $m_{j,i} = 0$.\footnote{This corresponds to adding the $i$th row to the $j$th row, which, as mentioned earlier, does not change the subspace spanned by the rows of $M$.}
\end{enumerate}
\item Apply a Hadamard gate to each of the first $d$ qubits.
\item For $i \in \{0, \ldots, d-1\}$, apply a $\mathrm{CNOT}$ with control qubit $n-i$ and target qubit $d-i$. Then apply a $\mathrm{CNOT}$ with control qubit $d-i$ and target qubit $n-i$. 
\end{enumerate}

Let $C$ denote the Clifford circuit described by the above process. 
The algorithm concludes by outputting $C$.
We apply $O(n)$ Clifford gates and do at most $O(d)$ row-sum operations per row. Thus, each iteration of Step 1 takes $O(nd)$ time. 
Since $d \leq \min(m, n)$, the overall running time is therefore $O(m n \cdot \min(m,n))$, due to the Gaussian elimination step at the beginning of the algorithm and the size of the circuit is at most $O(nd)$.

To prove correctness, we must argue that $C(A)= 0^{2n-d} \times \F_2^{d}$, or equivalently, that the algorithm above maps $M$ to a matrix as in \cref{eq:matrix-target}.

First, we show that Step 1 of the algorithm maps $M$ to a matrix of the form
\[
\begin{pmatrix}[cc|c]
   I & 0 & 0
\end{pmatrix},
\]
where $I$ is the $d \times d$ identity matrix.
It is clear that after the first iteration of Step 1 completes, $m_{1,1} =1$ and the remaining entries of the first row and column are $0$'s. 
By way of induction, assume that this is true after the first $i-1$ iterations, so that the resulting matrix looks as follows:

\vspace{0.5\baselineskip}
\begin{align}\label{eq:invariant-matrix}
\begin{pNiceArray}{cccccccc|cccccccc}[]
\Block[borders={bottom,right}]{4-4}{}
1 &&& & \Block[borders={bottom,right}]{4-4}{}
0 & \Cdots & \Cdots & 0 &
\Block[borders={bottom}]{4-8}{}
0 & \Cdots & \Cdots & \Cdots & \Cdots &\Cdots &\Cdots &0 \\
  & \Ddots && &  \Vdots &  & & \Vdots & \Vdots &   & &&&&&\Vdots  \\
  && \Ddots & &  \Vdots & &  & \Vdots &&&&&&&& \\ 
  &&& 1 &  0 & \Cdots & \Cdots & 0 & 0 & \Cdots &\Cdots &\Cdots&\Cdots&\Cdots&\Cdots& 0 \\  
\Block[borders={right}]{4-4}{}
0 & \Cdots & \Cdots & 0 &
\Block[borders={right}]{4-4}{} m_{i,i} & & & &
\Block[borders={right}]{4-4}{}
0 & \Cdots & &0 \\ 
\Vdots & & & \Vdots &&&&&\Vdots &  & & \Vdots \\ 
\Vdots &&  & \Vdots &&&&& \\ 
0 & \Cdots & \Cdots & 0 &&&&&0 &\Cdots & &0 \\ 
\CodeAfter
    \OverBrace[shorten,yshift=3.75pt]{1-1}{1-8}{n \text{ columns}}
    \UnderBrace[shorten,yshift=3.75pt]{8-1}{8-4}{i-1 \text{ columns}}
    \UnderBrace[shorten,yshift=3.75pt]{8-9}{8-12}{i-1 \text{ columns}}
\end{pNiceArray}.\\\nonumber
\end{align}
\vspace{\baselineskip}


\noindent In the top row, from left to right, the first block is the $(i-1) \times (i-1)$ identity matrix, then an $(i -1) \times (n-i+1)$ block of all $0$'s, and finally an $(i-1) \times 2n$ block of all $0$'s.
In the bottom row, from left to right, the first block is a $(d - i +1) \times (i-1)$ matrix of all $0$'s, then a $(d - i + 1) \times (n-i+1)$ block being processed by the algorithm, then a $(d-i+1) \times (i-1)$ block of all $0$'s, and finally a $(d-i+1)\times(n - i +1)$ block being processed by the algorithm.

We will argue that after the $i$th iteration of Step 1 finishes, the matrix will have the form of \cref{eq:invariant-matrix} but with $m_{i,i}=1$ and the rest of the $i$th row and $i$th column cleared to $0$. 
First, observe that the third block in the second row must be all $0$'s if the first block of the top row is the identity matrix because the subspace spanned by the rows is isotropic (and applying Clifford gates will not affect that
). 
It is also clear that the operations performed in the $i$th iteration will set $m_{i,i}=1$, set $m_{i,j} =0$ for $j \in \{i+1, \dots, 2n\}$, and set $m_{j,i}=0$ for $j \in \{i+1, \ldots, d\}$.
Therefore, we just need to argue that the $i$th iteration does not reintroduce $1$'s into the blocks of $0$'s or affect the $(i-1)\times(i-1)$ identity matrix in the first block of the first row.
Observe that neither of these can happen as long as Hadamard gates and $\mathrm{CNOT}$ gates are not applied to the first $i-1$ qubits in the $i$th iteration. Indeed, our algorithm does not apply any gates to the first $i-1$ qubits, so the structure of the matrix in \cref{eq:invariant-matrix} is preserved. 
Therefore, once Step 1 terminates, the resulting matrix will be a $d \times d$ identity matrix in the first block and the remaining entries of the matrix will be $0$'s.

The layer of Hadamard gates in Step 2 maps the $d\times d$ identity matrix to the right block of the matrix, i.e., 
\[
\begin{pmatrix}[c|cc]
   0 & I & 0
\end{pmatrix}.
\]
Finally, the $\mathrm{CNOT}$ gates in Step 3 move the identity matrix to the rightmost side of the tableau, matching the goal shown in \cref{eq:matrix-target}. This can be verified by explicit calculation. We note that the $\mathrm{CNOT}$ gates move the identity matrix by starting with the rightmost column and then proceeding leftward, which is critical for correctness when $d > n/2$.  
Hence, $C$ performs the desired mapping. \qedhere

\end{proof}

\section{On Subspaces with Large \texorpdfstring{$p_\psi$}{\emph{p}} or \texorpdfstring{$q_\psi$}{\emph{q}}-mass}

We prove a series of lemmas about subspaces of $\F_2^{2n}$ that form the backbone of our property testing and tomography algorithms. Most of these lemmas presuppose that a given subspace $S$ is ``heavy,'' in the sense that either $p_\psi$ or $q_\psi$ places a large probability mass on $S$. We begin by showing that heavy subspaces must be isotropic. Then, we show that if a sufficiently heavy subspace $S$ is $(n - t)$-dimensional, $\ket{\psi}$ is close in fidelity to a state of the form $C\ket{\varphi}\ket{x}$, where $C$ is a Clifford circuit determined by $S$, $\ket{x}$ is an $(n-t)$-qubit basis state, and $\ket{\varphi}$ is a $t$-qubit state.
Furthermore, we explain how to use \cref{lem:clifford-mapping-algorithm} to efficiently construct the Clifford circuit $C$ given $S$. With these lemmas in hand, the remaining work in this paper will lie largely in showing how to find such heavy subspaces efficiently.

\subsection{Subspaces With Large Mass Are Isotropic}

First, we relate the isotropicity of a subspace to the $q_\psi$ mass on its symplectic complement.

\begin{lemma}
\label{lem:isotropic}
    Let $S$ be a subspace of $\F_2^{2n}$ such that
    \[
    \sum_{x \in S^\perp}q_{\psi}(x) > \frac{5}{8}.
    \]
    Then $S$ is isotropic.
\end{lemma}
\begin{proof}
    Let $A \coloneqq \{x \in S: 2^n p_\psi(x) > 1/2\}$.
    By \cref{fact:M_half_commute}, every pair of elements in $A$ has symplectic product zero. 
    Furthermore, every pair of elements in $\langle A \rangle$ also has symplectic product zero by linearity of the symplectic product, showing that $\langle A \rangle$ is isotropic.
    Thus, if we can show that $\langle A \rangle = S$  then $S$ must also be isotropic.

   Suppose for a contradiction that $\abs{A} \le \frac{\abs{S}}{2}$. Then:
    \begin{align*}
    \sum_{x \in S^\perp}q_{\psi}(x) &= |S^\perp|\sum_{x \in S} p_\psi(x)^2 && (\text{\cref{thm:p_q_duality}}) \\
    &= |S^\perp|\left( \sum_{x\in A }p_\psi(x)^2 + \sum_{x \in S \setminus A} p_\psi(x)^2 \right)\\
    & \leq |S^\perp|\left(\frac{\abs{S}}{2}\cdot \frac{1}{4^n} + \frac{|S|}{2}\cdot \frac{1}{4\cdot 4^n}\right) \\
    &= \frac{|S|\cdot|S^\perp|}{4^n}\cdot \left(\frac{1}{2} + \frac{1}{8}\right)  \\
    &= \frac{5}{8},
    \end{align*}
    which contradicts the assumption of the lemma. So, $\abs{A} > \frac{\abs{S}}{2}$.
    Since a proper subspace of $S$ can have at most $\frac{\abs{S}}{2}$ elements, it follows that $\langle A \rangle = S$ and $S$ is isotropic. \qedhere
\end{proof}

Next, we prove an analogous statement with $p_\psi$ in place of $q_\psi$. For the next lemma, note that by \cref{thm:p_q_duality}, the assumption on $S$ is equivalent to
\[
\sum_{x \in S^\sympcomp}p_\psi(x) > \frac{3}{4},
\]
which is more directly comparable to \cref{lem:isotropic}. The proof is also similar.

\begin{lemma}
\label{lem:p-mass-isotropic}
    Let $S$ be a subspace of $\F_2^{2n}$ such that
    \[
    \sum_{x \in S}p_{\psi}(x) > \frac{3}{4}\frac{\abs{S}}{2^n}.
    \]
    Then $S$ is isotropic.
\end{lemma}
\begin{proof}
    Let $A \coloneqq \{x \in S: 2^n p_\psi(x) > 1/2\}$.
    By \cref{fact:M_half_commute}, every pair of elements in $A$ has symplectic product zero. 
    Furthermore, every pair of elements in $\langle A \rangle$ also has symplectic product zero by linearity of the symplectic product, showing that $\langle A \rangle$ is isotropic.
    Thus, if we can show that $\langle A \rangle = S$  then $S$ must also be isotropic.

    Suppose for a contradiction that $\abs{A} \le \frac{\abs{S}}{2}$. Then:
    \begin{align*}
    \sum_{x \in S}p_{\psi}(x) &=  \sum_{x\in A}p_\psi(x) + \sum_{x \in S \setminus A} p_\psi(x) \\
    & \leq \frac{\abs{S}}{2}\cdot \frac{1}{2^n} + \frac{\abs{S}}{2}\cdot \frac{1}{2\cdot 2^n} \\
    &= \frac{\abs{S}}{2^n}\cdot \left(\frac{1}{2} + \frac{1}{1}\right)  \\
    &= \frac{3}{4}\frac{\abs{S}}{2^n},
    \end{align*}
    which contradicts the assumption of the lemma. So, $\abs{A} > \frac{\abs{S}}{2}$.
    Since a proper subspace of $S$ can have at most $\frac{\abs{S}}{2}$ elements, it follows that $\langle A \rangle = S$ and therefore $S$ is isotropic.
\end{proof}

\subsection{Product State Structure}\label{subsec:prod-state-structure}
We now establish the critical relation between $p_\psi$-mass on isotropic subspaces and stabilizer dimension. As noted earlier, we show that if an $(n-t)$-dimensional isotropic subspace has large  $p_\psi$-mass, then there is a Clifford circuit $C$ (that can be constructed efficiently) that approximately maps $\ket{\psi}$ to a product state $\ket{\varphi}\ket{x}$, where $\ket{x}$ is an $(n-t)$-qubit basis state and $\ket{\varphi}$ is an arbitrary $t$-qubit state. 

Special cases of this relation have appeared implicitly before. In particular, if the $p_\psi$ mass on the subspace is \emph{maximal} (i.e., exactly $\frac{1}{2^t}$), then this can be viewed as a restatement of the fact that stabilizer codes can be encoded using Clifford circuits \cite{gottesman2006quantum}. This special case also essentially follows from \cite[Theorem 2]{leone-stabilizer-nullity}. The main contribution of this subsection is to show that this mapping is robust to error.

To begin, we prove a relation between sums over basis state projections and sums of Pauli-$Z$ strings.



\begin{proposition}\label{prop:sum-over-stabilizer-basis}
    \[\sum_{x \in \F_2^{k} }\ketbra{x}{x} \otimes \ketbra{x}{x} = \frac{1}{2^k} \sum_{x \in \F_2^k} Z^x \otimes Z^x,\]
    where $Z^x \coloneqq Z^{x_1} \otimes \dots \otimes Z^{x_{k}}$.
\end{proposition}
\begin{proof}
We use the fact that for any $x \in \F_2^k$, $\ketbra{x}{x} = 2^{-k} \sum_{a \in \F_2^k} Z^a (-1)^{a \cdot x}$.\footnote{This is essentially the Weyl expansion of $\ketbra{x}{x}$. Since the expression only involves $I$'s and Pauli $Z$'s, we can write the expansion in this simplified form.}
    \begin{align*}
        \sum_{x \in \F_2^k}\ketbra{x}{x} \otimes \ketbra{x}{x}
        &= \sum_{x \in \F_2^k} \frac{1}{4^k}\left(\sum_{a \in \F_2^k} Z^a (-1)^{a \cdot x}\right) \otimes \left(\sum_{b \in \F_2^k} Z^b (-1)^{b \cdot x}\right)\\
        &= \frac{1}{4^k}\sum_{a, b \in \F_2^k} Z^a \otimes Z^b \sum_{x \in \F_2^k} (-1)^{(a + b) \cdot x}\\
        &= \frac{1}{2^k} \sum_{a \in \F_2^k} Z^a \otimes Z^a.\qedhere 
    \end{align*}
\end{proof}

Next, we prove that the $p_\psi$-mass on an isotropic subspace has a nice operational interpretation.

\begin{proposition}\label{prop:collision}
    Let $\ket{\psi}$ be an $n$-qubit state, and let $S = 0^{n+t} \times \F_2^{n-t}$. 
    Upon measuring the last $n-t$ qubits in the computational basis on $2$ copies of $\ket{\psi}$, the probability of observing the same string $x \in \F_2^{n-t}$ twice (i.e., the collision probability) is 
    \[ 2^t \sum_{x \in S} p_\psi(x). \]
\end{proposition}
\begin{proof}
The probability of observing some $x \in \F_2^{n-t}$ twice is 
\begin{align*}
\sum_{x \in \F_2^{n-t}} \bra{\psi} \left( I^{\otimes t} \otimes \ketbra{x}{x}\right)  \ket{\psi}^2
&= \bra{\psi}^{\otimes{2}} \left(\sum_{x \in \F_2^{n-t}} I^{\otimes t} \otimes \ketbra{x}{x} \otimes I^{\otimes t} \otimes \ketbra{x}{x} \right) \ket{\psi}^{\otimes 2}\\
&= \bra{\psi}^{\otimes{2}}  \left( \frac{1}{2^{n-t}}\sum_{x \in \F_2^{n-t}} I^{\otimes t} \otimes Z^x \otimes I^{\otimes t} \otimes Z^x \right) \ket{\psi}^{\otimes 2}\\
&= \frac{1}{2^{n-t}} \sum_{x \in \F_2^{n-t}} \braket{\psi|I^{\otimes t} \otimes Z^x |\psi}^2\\
&= 2^t \sum_{x \in S} p_\psi(x). 
\end{align*}
The third step follows from \cref{prop:sum-over-stabilizer-basis} by treating the $I^{\otimes t}$ as constants.
\end{proof}

Now we prove the main lemma of this subsection for the special case where the isotropic subspace is a subset of $\{I,Z\}^{\otimes n}$. 
\begin{lemma}
\label{lem:product-state-approximation}
    Let $S = 0^{n+t} \times \F_2^{n-t}$, and suppose that 
    \[\sum_{x \in S} p_\psi(x) \geq \frac{1 - \eps}{2^t}.\] 
    Then there exists an $(n-t)$-qubit computational basis state $\ket{x}$ and a $t$-qubit quantum state
    \[
     \ket{\varphi} \coloneqq \frac{(I \otimes \bra{x}) \ket{\psi}}{\norm{(I \otimes \bra{x}) \ket{\psi}}_2},\footnote{This is to say that $\ket{\varphi}$ is obtained by postselecting on measuring the last $n - t$ qubits of $\ket{\psi}$ to be $\ket{x}$.}
    \]
    such that
    the fidelity between $\ket{\varphi}\ket{x}$ and $\ket{\psi}$ is at least $1 - \eps$.  
    
\end{lemma}
\begin{proof}
    We can always write $\ket{\psi} = \sum_{x \in \F_2^{n-t}} \alpha_x \ket{\varphi_x}\ket{x}$ where $\sum_{x \in \F_2^{n-t}} \abs{\alpha_x}^2 = 1$.
    If we can show that \[\max_{x \in \F_2^{n-t}} \abs{\bra{\psi} (\ket{ \varphi_x}\ket{x})}^2 \geq 1- \eps,\] then we are done, by taking $\ket{\varphi} = \ket{\varphi_x}$. First, 
    \begin{align*}
        \max_{x \in \F_2^{n-t}} \abs{\bra{\psi} (\ket{ \varphi_x}\ket{x})}^2
        &= \max_{x \in \F_2^{n-t}} \abs{\alpha_x}^2 \\
        &= \max_{x \in \F_2^{n-t}} \abs{\alpha_x}^2 \cdot \sum_{x \in \F_2^{n-t}} \abs{\alpha_x}^2\\
        &\geq \sum_{x \in \F_2^{n-t}} \abs{\alpha_x}^4. 
    \end{align*}
    
    Observe that $\sum_x \abs{\alpha_x}^4$ is precisely the collision probability when measuring the last $n-t$ qubits of $\ket{\psi}$ in the computational basis. Hence, by \cref{prop:collision}, 
    \begin{align*}
        \sum_{x \in \F_2^{n-t}} \abs{\alpha_x}^4
        &= 2^t \sum_{x \in S} p_\psi(x) \geq 1 - \eps.\qedhere   
    \end{align*}
    \end{proof}

    Finally, we generalize the previous lemma using the Clifford mapping algorithm from \Cref{lem:clifford-mapping-algorithm}. 
    \begin{lemma}
        \label{lem:product-state-approximation-general}
        Let $S$ be an isotropic subspace of dimension $n-t$, and suppose that 
    \[\sum_{x \in S} p_\psi(x) \geq \frac{1 - \eps}{2^t}.\] 
    Then there exists a state $\ket{\hat{\psi}}$ with $S \subseteq \weyl(\ket{\hat{\psi}})$ such that the fidelity between $\ket{\hat{\psi}}$ and $\ket{\psi}$ is at least $1 - \eps$. 
    
    In particular, $\ket{\hat{\psi}} = C^\dagger \ket{\varphi}\ket{x}$, where $\ket{x}$ is an $(n-t)$-qubit basis state, \[
     \ket{\varphi} \coloneqq \frac{(I \otimes \bra{x}) C\ket{\psi}}{\norm{(I \otimes \bra{x}) C\ket{\psi}}_2}
    \] is a $t$-qubit quantum state, and $C$ is the Clifford circuit mapping $S$ to $0^{n+t} \times \F_2^{n-t}$ described in \cref{lem:clifford-mapping-algorithm}.
    \end{lemma}
    
    \begin{proof}
    Define $\ket{\phi} \coloneqq C\ket{\psi}$.
   Then, by \cref{fact:clifford-permutes-p-mass}, 
   \[
        \sum_{x \in S}p_{\psi}(x) = \sum_{x \in C(S)}p_{\phi}(x) \geq \frac{1-\eps}{2^t}.
   \]
   Therefore, by \cref{lem:product-state-approximation}, $C\ket\psi$ is $(1-\eps)$-close in fidelity to a state $\ket{\varphi}\ket{x}$, where $\ket{x}$ is an $(n-t)$-qubit basis state and
   \[
     \ket{\varphi} \coloneqq \frac{(I \otimes \bra{x}) C\ket{\psi}}{\norm{(I \otimes \bra{x}) C\ket{\psi}}_2}
    \]
   is a $t$-qubit quantum pure state.
   Since fidelity is unitarily invariant, the fidelity between $\ket{\psi}$ and $\ket{\hat{\psi}} = C^\dagger\ket{\varphi}\ket{x}$ is also at least $1-\eps$. 
   Clearly, $C(S) \subseteq \weyl(\ket{\varphi} \ket{x})$, and therefore, by \cref{fact:clifford-permutes-p-mass}, $S \subseteq \weyl(\ket{\hat \psi})$.
    \end{proof}

\section{Property Testing Stabilizer Dimension}\label{sec:property-testing}

As a first application, we present an efficient algorithm for property testing stabilizer dimension. 
Recall that a property tester for a class $\calQ$ of quantum states takes copies of a state $\ket\psi$ as input and determines whether $\ket\psi \in \calQ$ or $\ket\psi$ is $\eps$-far from all such states (according to some measure of distance), promised that one of these is the case. 
Our algorithm efficiently tests whether an input state has stabilizer dimension at least $k$ or has fidelity less than $1-\eps$ with all such states. 

\begin{algorithm}[H]
\caption{Property Testing Stabilizer Dimension}\label{alg:property-testing}
\SetKwInOut{Promise}{Promise}
\KwInput{$\frac{16n+8\log(1/\delta)}{\eps}$ copies of $\ket\psi$, $k \in [n]$, $\eps \in (0,3/8)$, and $\delta \in (0,1]$}
\Promise{$\ket\psi$ has stabilizer dimension at least $k$ or is $\eps$-far in fidelity from all such states} 
\KwOutput{$1$ if $\ket\psi$ has stabilizer dimension at least $k$, $0$ otherwise, with probability at least $1 - \delta$}

Perform Bell difference sampling to draw $\frac{4n+2\log(1/\delta)}{\eps}$ samples from $q_\psi$. Denote the span of the samples by $S^\sympcomp$. 

Let $\hat{k} = \dim S = 2n - \dim S^\sympcomp$.

\Return{$1$ if $\hat{k} \geq k$ and $0$ otherwise.}
    
\end{algorithm}

\begin{theorem}
\label{thm:property-testing-alg}
Let $\ket{\psi}$ be an $n$-qubit quantum state. 
For $\eps \in (0, 3/8)$, \Cref{alg:property-testing} determines whether $\ket{\psi}$ has stabilizer dimension at least $k$ or has fidelity at most $1-\eps$ with all such states, promised that one of these is the case.
\sloppy
The algorithm uses $\frac{16n + 8\log(1/\delta)}{\eps}$ copies of $\ket{\psi}$ and $O\left(\frac{n^3 + n^2\log(1/\delta)}{\eps}\right)$ time, and succeeds with probability at least $1-\delta$. 
\end{theorem}

\begin{proof}
    First, suppose that $\ket\psi$ has stabilizer dimension at least $k$. By \cref{fact:p-support-sympcomp}, $q_\psi$ is supported on a subspace of dimension at most $2n-k$. 
    So, no matter what, $\dim S^\sympcomp$ will never exceed $2n - k$, and therefore $\hat{k}$ will always be at least $k$. Hence, \Cref{alg:property-testing} accepts with probability $1$.

    Now suppose that $\ket{\psi}$ has fidelity less than $1 - \eps$ with every state of stabilizer dimension at least $k$. By \cref{prop:q-lower-bounds-p} and \Cref{lem:sampling}, with probability at least $1 - \delta$, the subspace $S$ satisfies
    \[
    \sum_{x \in S^\sympcomp} p_\psi(x) \geq \sum_{x \in S^\sympcomp} q_\psi(x) \geq 1-\eps.
    \]
    Assuming this occurs, applying \Cref{thm:p_q_duality} gives
    \[
    \sum_{x \in S} p_\psi(x) = \frac{\abs{S}}{2^n}\sum_{x\in S^\sympcomp} p_\psi(x) \ge \frac{1 - \eps}{2^{n - \hat{k}}}.
    \]
    Since $\sum_{x \in S^\sympcomp} q_\psi(x) \geq 1-\eps > 5/8$, $S$ is isotropic by \cref{lem:isotropic}. As such, by
    \Cref{lem:product-state-approximation-general}, there exists a state of stabilizer dimension at least $\hat{k}$ that has fidelity at least $1-\eps$ with $\ket\psi$. By assumption, we must have $\hat{k} < k$, and thus \Cref{alg:property-testing} rejects with probability at least $1 - \delta$.
    
    Overall, we find that in either case, with probability at least $1-\delta$, the algorithm succeeds.
    It remains to bound the runtime.
    \Cref{lemma:compute-symplectic-complement} tells us that computing a basis for $S^\sympcomp$ requires $O(mn^2)$ time, where $m = O\left(\frac{n + \log(1/\delta)}{\eps}\right)$ is the number of Bell difference samples taken by the algorithm. This dominates the running time.
\end{proof}

\section{Tomography Reduces to Finding Heavy Subspaces}

In this section, we describe a reduction from tomography on states $\ket{\psi}$ with large stabilizer dimension to finding any subspace $S \subseteq \F_2^{2n}$ that approximates $\weyl(\ket\psi)$ (in a sense that is made explicit below). This reduction uses only single-copy measurements. In the later sections, we will give algorithms to find such subspaces $S$ using either entangled or single-copy measurements.

\subsection{Pure State Tomography}
Our reduction requires a general pure state tomography algorithm as a subroutine. The complexity of our reduction will be parameterized by the complexity of the tomography algorithm used. To describe these runtimes succinctly, it will be convenient to introduce the following notation.

\begin{definition}[Single-copy pure state tomography copy and time complexities]\label{def:tomography-complexity}
Let $N_{n, \eps, \delta}$ and $M_{n, \eps, \delta}$ be the sample and time complexities, respectively, of pure state tomography on $n$ qubits to trace distance $\eps$ with failure probability at most $\delta$, using only single-copy measurements.
\end{definition}

For context, it was shown by \cite{franca_et_al:LIPIcs.TQC.2021.7} that there exists a single-copy pure state tomography procedure that uses $O\left( 2^n n \log(1/\delta)\eps^{-4}\right)$ copies and $O\left( 4^n n^3 \log(1/\delta)\eps^{-5}\right)$ time to output a state $\ket{\hat{\psi}}$ that is $\eps$-close to $\ket{\psi}$ in trace distance with probability at least $1-\delta$. To our knowledge, this is the state-of-the-art when it comes to scaling in the dimension $d = 2^n$ of the system (though, it does not match the best-known lower bound of $\Omega(d)$). 

\begin{theorem}[\cite{franca_et_al:LIPIcs.TQC.2021.7}]\label{thm:tomography}
Given access to copies of an $n$-qubit pure state $\ket{\psi}$,
there is a single-copy algorithm that uses $O\left( 2^n n \log(1/\delta)\eps^{-4}\right)$ copies, $O\left( 4^n n^3 \log(1/\delta)\eps^{-5}\right)$ time, and outputs a state $\ket{\hat{\psi}}$ that is $\eps$-close to $\ket{\psi}$ in trace distance with probability at least $1-\delta$. The algorithm only requires applying random Clifford circuits and classical post-processing. 
\end{theorem}

Alternatively, it was shown by \cite{aaronson2023efficient} that there exists a single-copy pure state tomography procedure that uses $O\left( 16^n \log(1/\delta)\eps^{-2}\right)$ copies and $O\left( 32^n \log(1/\delta)\eps^{-2}\right)$ time, sacrificing dependence on $n$ for a better dependence on $\eps$.

\begin{theorem}[{\cite[Section 5]{aaronson2023efficient}}]\label{thm:tomography-option2}
Given access to copies of an $n$-qubit pure state $\ket{\psi}$,
there is a single-copy algorithm that uses $O\left( 16^n \log(1/\delta)\eps^{-2}\right)$ copies, $O\left( 32^n \log(1/\delta)\eps^{-2}\right)$ time, and outputs a state $\ket{\hat{\psi}}$ that is $\eps$-close to $\ket{\psi}$ in trace distance with probability at least $1-\delta$. The algorithm only requires applying Clifford circuits and classical post-processing. 
\end{theorem}

In either case, we have the upper bound $M_{n,\eps,\delta}, N_{n,\eps,\delta} \le \poly(2^n, 1/\eps, \log(1/\delta))$.



\subsection{Reduction Algorithm}

We are now ready to describe the reduction algorithm. Recall \cref{lem:product-state-approximation-general}, which shows that given a large isotropic subspace that accounts for a large fraction of the $p_\psi$ mass, there exists a Clifford circuit that approximately maps $\ket{\psi}$ to a product $\ket{\varphi}\ket{x}$ of a state $\ket{\varphi}$ on few qubits and a computational basis state $\ket{x}$. Our algorithm below amounts to computing this Clifford circuit $C$ and then performing tomography on the states $\ket{\varphi}$ (via the general tomography algorithm) and $\ket{x}$ (by measuring in the computational basis).

\begin{algorithm}[H]
\caption{Learning reduction using single-copy measurements}\label{alg:reduction}
\SetKwInOut{Promise}{Promise}
\KwInput{Black-box access to copies of $\ket\psi$, $\eps, \delta \in (0,1)$, and $S \subseteq \F_2^{2n}$}
\Promise{$S$ is a subspace such that $\sum_{x \in S} p_\psi(x) \geq \left(1-\frac{\eps^2}{4}\right)\frac{|S|}{2^n}$ and $\dim(S) \geq n-t$}
\KwOutput{A classical description of $\ket{\hat{\psi}}$ such that $\tracedistance{\ket \psi, \ket{\hat{\psi}}} \leq \eps$ with probability at least $1-\delta$}

Let $\hat{t}$ be such that $\dim(S) = n-\hat{t}$ and $\hat{t} \leq t$.
\label{step:zero-reduction}

Use the algorithm in \cref{lem:clifford-mapping-algorithm} to construct a Clifford circuit $C$ that maps $S$ to $0^{n+\hat{t}} \times  \F_2^{n - \hat{t}}$.
\label{step:one-reduction}

Measure the last $n-\hat{t}$ qubits of $24 \log(3/\delta)$ copies of $C\ket\psi$. Let $\hat{x} \in \F_2^{n-\hat{t}}$ be the majority result.
\label{step:two-reduction}

\For{$i \gets 1$ \KwTo $\frac{4}{3}N_{\hat{t},\frac{\eps}{2},\frac{\delta}{3}} + \frac{8}{9} \log(3/\delta)$}{
\label{step:three-reduction}

Measure the last $n-\hat{t}$ qubits of a copy of $C\ket\psi$ and let $x_i \in \{0, 1\}^{n-\hat{t}}$ be the outcome. 
\label{step:four-reduction}

\uIf{$x_i = \hat{x}$}{
Use the state on the first $\hat{t}$ qubits as the input to a pure state single-copy tomography algorithm with accuracy $\eps/2$ in trace distance and failure probability at most $\delta/3$.
}
\label{step:five-reduction}
\Else{ Discard the state.}
}
Let $\ket{\hat\varphi}$ be the classical description output by the pure state single-copy tomography algorithm.
\label{step:six-reduction}

\Return{$\ket{\hat{\psi}} = C^\dagger\ket{\hat{\varphi}}\ket{\hat{x}}$.}
\label{step:seven-reduction}
\end{algorithm}

To aid intuition, note that if $\ket{\psi}$ has stabilizer dimension $n - t$ and $S = \weyl(\ket \psi)$, then the condition of the algorithm is satisfied with $\eps = 0$. So, the algorithm may be understood as showing how to perform tomography on $\ket{\psi}$ given a sufficiently good approximation $S$ to $\weyl(\ket \psi)$.

\begin{theorem}
\label{thm:learn-weyl-psi-reduction}
Let $\ket\psi$ be an $n$-qubit state, and suppose there exists a subspace $S \subseteq \F_2^{2n}$ that satisfies the following conditions: 
\begin{enumerate}[label=\rm{(\roman*)}]
    \item\label{item:dimension} $\dim(S) \geq n-t$, and
    \item\label{item:approx-weyl} $\sum_{x \in S} p_\psi(x) \geq \left(1-\frac{\eps^2}{4}\right)\frac{\abs{S}}{2^n}$.
\end{enumerate}
Then, given $S$, \cref{alg:reduction} outputs a classical description of $\ket{\hat{\psi}}$ such that $\tracedistance{\ket\psi, \ket{\hat{\psi}}} \leq \eps$ with probability at least $1-\delta$ using $\frac{4}{3}N_{t, \frac{\eps}{2}, \frac{\delta}{3}} + \frac{224}{9} \log(3/\delta)$ copies of $\ket\psi$, and 
    $M_{t, \frac{\eps}{2}, \frac{\delta}{3}} + O\left(n^2 \left(N_{t, \frac{\eps}{2}, \frac{\delta}{3}} + \log(1/\delta)\right)\right)$  time,\footnote{As seen by \cref{thm:tomography,thm:tomography-option2}, we can usually simplify this to just $O\left(M_{t, \frac{\eps}{2}, \frac{\delta}{2}}\right)$ in practice.} while performing only single-copy measurements.
\end{theorem}

\begin{proof}
    \cref{lem:p-mass-isotropic} implies that a subspace $S$ satisfying \cref{item:approx-weyl} must be isotropic, because $\eps < 1$. The algorithm from \cref{lem:clifford-mapping-algorithm} produces a Clifford circuit $C$ that maps $S$ to $0^{n+\hat t} \times \F_2^{n-\hat t}$. 
     By \Cref{lem:product-state-approximation-general}, there exists an $(n-\hat t)$-qubit basis state $\ket{x}$ and a $\hat{t}$-qubit state
     \[
     \ket{\varphi} \coloneqq \frac{(I \otimes \bra{x}) C \ket{\psi}}{\norm{(I \otimes \bra{x}) C \ket{\psi}}_2},
     \]
     such that the fidelity between $C \ket\psi$ and $\ket{\varphi}\ket{x}$ is at least $1-\eps^2/4$, implying that
     \begin{align}\label{eq:first-td-bound}
     \tracedistance{\ket{\psi}, C^\dagger \ket{ \varphi}\ket{x}} \leq \frac{\eps}{2}. 
     \end{align}

In \cref{step:two-reduction}, for $1 \le i \le m = 24\log(3/\delta)$, let $x_i$ denote the $i$th ($n-\hat{t}$)-bit measurement outcome upon measuring the last $n-\hat t$ qubits of $C\ket\psi$ in the computational basis.
Define the indicator random variable $X_i$ as 
\begin{align}\label{eq:random-variable-post-selection}
X_i = \begin{cases}
1 & \text{if $x_i = x$.} \\
0 & \text{otherwise.} \\
\end{cases}
\end{align}
Because the fidelity between $C\ket{\psi}$ and $\ket{\varphi}\ket{x}$ is at least $1 - \eps^2/4$ and $\eps \in (0,1)$, $\Pr[X_i = 1] \geq 3/4$ and $\mu \coloneqq \E[\sum_i X_i] \geq 3m/4$. 
By a Chernoff bound (\cref{fact:chernoff}),

\begin{align*}
\Pr\left[\sum_{i = 1}^m X_i \leq \frac{m}{2}\right]
&\leq \Pr\left[ \sum_{i=1}^m X_i \leq \left(1-\frac{1}{3}\right) \mu  \right] \\
&\le \exp\left( -\frac{\mu}{18} \right)  \\
&\leq \exp\left( -\frac{ m}{24}\right)\\
&= \delta/3.
\end{align*}
Hence, over half of the $m$ samples will be $x$ with probability at least $1-\delta/3$, which means that we will have $\hat{x} = x$ with probability $1 - \delta/3$.

Assume now that we have successfully learned $\hat{x} = x$.
We then need to perform pure state tomography on $\ket{\varphi}$. Note that the definition of $\ket{\varphi}$ means that $\ket{\varphi}$ is the state on $\hat{t}$ qubits conditioned on measuring $\ket{x}$ on the last $n - \hat{t}$ qubits of $C\ket{\psi}$. In \cref{step:three-reduction}, we need to repeat this postselection enough times until we have $N_{\hat{t}, \frac{\eps}{2},\frac{\delta}{3}}$ copies of $\ket{\varphi}$ to feed into the tomography algorithm.

In \cref{step:four-reduction}, for $1 \le i \le m = \frac{4}{3}N_{\hat{t},\frac{\eps}{2},\frac{\delta}{3}} + \frac{8}{9} \log(3/\delta)$, define the random variable $Y_i$ by
\begin{align}\label{eq:random-variable-post-selection-2}
Y_i = \begin{cases}
1 & \text{if $x_i = x$.} \\
0 & \text{otherwise.} \\
\end{cases}
\end{align}
As with the $X_i$'s, because the fidelity between $C\ket{\psi}$ and $\ket{\varphi}\ket{x}$ is at least $1 - \eps^2/4$ and $\eps \in (0,1)$, $\Pr[Y_i = 1] \geq 3/4$ and $\mu \coloneqq \E[\sum_i Y_i] \geq 3m/4$.
By Hoeffding's inequality (\cref{fact:hoeffding}),

\begin{align*}
    \Pr\left[\sum_{i=1}^m Y_i < N_{\hat{t}, \frac{\eps}{2},\frac{\delta}{3}}\right]
    \leq \Pr\left[\sum_{i=1}^m Y_i - \mu \leq N_{\hat{t}, \frac{\eps}{2},\frac{\delta}{3}} - \frac{3}{4}m\right]
    \leq \exp\left(-\frac{2}{m}\left(N_{\hat{t}, \frac{\eps}{2},\frac{\delta}{3}}-\frac{3}{4}m\right)^2\right).
\end{align*}
Let $\gamma = \frac{8}{9}\log(\delta/3)/N_{\hat{t}, \frac{\eps}{2},\frac{\delta}{3}}$ so that we may write $m = \left(\frac{4}{3} + \gamma\right)N_{\hat{t}, \frac{\eps}{2},\frac{\delta}{3}}$. Then the above expression is upper bounded by
\[
\exp\left(-\frac{9}{8}\frac{\gamma^2}{\frac{4}{3} + \gamma}N_{\hat{t}, \frac{\eps}{2},\frac{\delta}{3}}\right) < \exp\left(-\frac{9}{8}\gamma N_{\hat{t}, \frac{\eps}{2},\frac{\delta}{3}}\right) \le \delta/3.
\]
This is to say that with probability at least $1 - \delta/3$, the algorithm postselects enough copies of $\ket{\varphi}$ in \cref{step:five-reduction} to successfully run pure state tomography on $\ket{\varphi}$.

By assumption, the pure state tomography algorithm in \cref{step:six-reduction} fails with probability at most $\delta / 3$. Assuming the tomography succeeds in outputting a state $\ket{\hat{\varphi}}$ satisfying $d_\tr\left(\ket{\varphi}, \ket{\hat{\varphi}}\right) \leq \eps/2$,
the returned state $C^\dagger\ket{\hat{\varphi}}\ket{\hat{x}}$ satisfies
\begin{align*}
    d_\tr(\ket{\psi}, C^\dagger\ket{\hat{\varphi}}\ket{\hat{x}}) 
    &\leq d_\tr(\ket{\psi}, C^\dagger\ket{\varphi}\ket{x}) + d_\tr(\ket{\varphi}\ket{x}, \ket{\hat{\varphi}}\ket{\hat{x}}) \\
    &\leq \frac{\eps}{2} + d_\tr(\ket{\varphi}\ket{x}, \ket{ \hat{\varphi}}\ket{\hat{x}}) \\
    &\leq \eps.
\end{align*}
Above, first step follows by the triangle inequality and the fact that trace distance is unitarily invariant. The second step follows from \cref{eq:first-td-bound}, and the final step follows assuming $\hat{x} = x$ and $d_\tr\left(\ket{\varphi}, \ket{\hat{\varphi}}\right) \leq \eps/2$.

By applying a union bound over the \say{bad} events (namely, \cref{step:two-reduction,step:five-reduction,step:six-reduction} of the algorithm), we have that the overall success probability of the algorithm is at least $1-\delta$.

To conclude, we bound the copy and time complexities of the algorithm.
Since $\hat{t} \leq t$, we use $O\left(N_{t, \frac{\eps}{2}, \frac{\delta}{3}}  + \log(1/\delta)\right)$ copies in the for loop of \cref{step:three-reduction}, which dominates the sample complexity.
Producing these states requires time $O(n^2)$ per state because of the size of the Clifford circuit from \cref{lem:clifford-mapping-algorithm}.
The runtime of the pure state single-copy tomography algorithm in \cref{step:six-reduction} otherwise dominates the running time. The overall time complexity is therefore
\[
O\left(n^2 \left(N_{t, \frac{\eps}{2},\frac{\delta}{3}} + \log(1/\delta)\right)\right)  + M_{t, \frac{\eps}{2}, \frac{\delta}{3}}.\qedhere 
\]
\end{proof}

Note that the dependence on $t$ comes entirely from the pure state tomography algorithm (e.g., \cref{thm:tomography} or \cref{thm:tomography-option2}) in \cref{step:six-reduction} of \cref{alg:reduction}. Hence, one can upgrade this part of the algorithm with improved pure state single-copy tomography algorithms when/if they are discovered. 

We also note that if one is allowed a quantum memory, a constant factor improvement in sample complexity is possible by saving the first $\hat{t}$ qubits in \cref{step:two-reduction}. Once $x$ has been determined via majority, these copies of $\ket{\varphi}$ can be fed into the pure state tomography algorithm.

\section{Learning the Unsigned Stabilizer Group Using Bell Measurements}
\label{sec:learning_weyl_bell}

In this section, we present the first of two algorithms for learning a subspace that satisfies the conditions of \cref{thm:learn-weyl-psi-reduction}. The algorithm in this section uses Bell difference sampling and is closely related to the property testing algorithm for stabilizer dimension (\cref{alg:property-testing}).

\begin{algorithm}[H]
\caption{Approximating $\weyl(\ket \psi)$ using Bell measurements}\label{alg:learning_weyl_bell}
\SetKwInOut{Promise}{Promise}
\KwInput{Black-box access to copies of $\ket\psi$ and $\eps, \delta \in (0, 1)$}
\Promise{$\ket\psi$ has stabilizer dimension at least $n-t$}
\KwOutput{A subspace $S \subset \F_2^{2n}$ of dimension at least $n-t$ such that $\sum_{x \in S} p_\psi(x) \geq \left(1-\frac{\eps^2}{4}\right)\frac{\abs{S}}{2^n}$ with probability at least $1 - \delta$}

Perform Bell difference sampling to draw $\frac{8\log(1/\delta)+16n}{\eps^2}$ samples from $q_\psi$. Denote the span of the samples by $S^\sympcomp$.
\label{step:one_bell} 

Use the algorithm in \cref{lemma:compute-symplectic-complement} to compute the symplectic complement $S$ of the subspace spanned by the samples. 
\label{step:two_bell}

\Return{$S$}
\end{algorithm}

Showing that \cref{alg:learning_weyl_bell} satisfies the conditions of \cref{thm:learn-weyl-psi-reduction} proceeds similarly to the analysis of \cref{alg:property-testing} in \cref{thm:property-testing-alg}.

\begin{theorem}\label{thm:bell-copy-main}
    Let $\ket{\psi}$ be an $n$-qubit quantum state with stabilizer dimension at least $n-t$.
    Given copies of $\ket{\psi}$ as input, \cref{alg:learning_weyl_bell} succeeds at outputting a subspace $S$ that satisfies the conditions of \cref{thm:learn-weyl-psi-reduction} with probability at least $1-\delta$.
    The algorithm uses
    \[
    \frac{32\log(1/\delta)+64n}{\eps^2}
    \]
    samples and
    \[
    O\left(\frac{n^2\log(1/\delta) + n^3}{\eps^2}\right),
    \]
    time.
\end{theorem}

\begin{proof}
    By \Cref{fact:p-support-sympcomp}, the support of $q_\psi$ is a subspace of dimension at most $n + t$, so $\dim S^\sympcomp \le n + t$ and therefore $\dim S \ge n - t$, by \cref{fact:sympcomp}. This establishes \cref{item:dimension}.

        By \Cref{lem:sampling}, except with probability at most $\delta$, the Bell difference sampling phase of the algorithm in \Cref{step:one_bell} finds a subspace $S^\sympcomp \subseteq \F_2^{2n}$ such that 
    \[
    \sum_{x \in S^\sympcomp} q_\psi(x) \geq 1-\frac{\eps^2}{4}.
    \]
    From \cref{prop:q-lower-bounds-p}, we know that 
    \[
    \sum_{x \in S^\sympcomp} p_\psi(x) \geq 1-\frac{\eps^2}{4}.
    \]
    In the next step, the algorithm computes $S$.
    Using \Cref{thm:p_q_duality}, we have that
    \[
    \sum_{x \in S} p_\psi(x) = \dfrac{\abs{S}}{2^n} \sum_{x \in S^\sympcomp} p_\psi(x) \geq \left(1-\frac{\eps^2}{4}\right)\frac{|S|}{2^n},
    \]
    which establishes \cref{item:approx-weyl}.

    The upper bound on sample complexity is immediate from the fact that each Bell difference sample requires $4$ copies of $\ket{\psi}$, so it only remains to bound the time complexity. The Bell difference sampling phase takes $O(n)$ time per sample, while \cref{lemma:compute-symplectic-complement} implies that computing the symplectic complement takes time
    \[
    O\left(\frac{n^2\log(1/\delta) + n^3}{\eps^2}\right),
    \]
    which dominates the runtime.
\end{proof}

Combining with \cref{alg:reduction} gives us the following.

\begin{corollary}\label{cor:bell-copy-main}
    Let $\ket{\psi}$ be an $n$-qubit quantum state with stabilizer dimension at least $n-t$.
    Given copies of $\ket{\psi}$ as input, the combination of \cref{alg:learning_weyl_bell,alg:reduction} outputs a classical description of $\ket{\hat{\psi}}$ such that $\tracedistance{\ket{\psi}, \ket{\hat{\psi}}} \leq \eps$ with probability at least $1-\delta$.
    The algorithm uses
    \[
    \frac{32\log(2/\delta)+64n}{\eps^2} + \frac{4}{3}N_{t, \frac{\eps}{2}, \frac{\delta}{6}} + \frac{224}{9}\log(6/\delta)
    \]
    samples and
    \[
    O\left(\frac{n^2\log(1/\delta) + n^3}{\eps^2} + n^2 \cdot N_{t, \frac{\eps}{2}, \frac{\delta}{6}} \right) + M_{t, \frac{\eps}{2}, \frac{\delta}{6}},
    \]
    time.
\end{corollary}
\begin{proof}
    By \cref{thm:learn-weyl-psi-reduction,thm:bell-copy-main}, each algorithm succeeds with probability at least $1-\delta/2$. By the union bound, the total failure probability is at most $\delta$.
\end{proof}

\section{Computational Difference Sampling}

In this section, we introduce and develop the theory of \emph{computational difference sampling} and connect it to the characteristic distribution $p_\psi$ 
(\cref{subsec:weyl-expansion-and-bell-diff-sampling}).

\begin{definition}[Computational difference sampling]
    Computational difference sampling a quantum state $\ket\psi$ corresponds to the following quantum measurement. Measure $\ket\psi$ in the computational basis to get outcome $x \in \F_2^n$, measure an additional copy of $\ket\psi$ in the computational basis to get another outcome $y\in \F_2^{n}$, and output $(x + y, 0^n) \in \Xs$.  
    Let $r_\psi(a)$ denote the probability of sampling $a \in \Xs$ from this process.
\end{definition}

One can view computational difference sampling as a distribution over $\F_2^n$, but for our purposes, it will be most convenient to identify $\F_2^n$ with $\Xs$, which corresponds to the set of Pauli-$X$ Weyl operators.

Computational difference sampling only uses single-copy measurements and is the core primitive used in our single-copy learning algorithm.
Although it is weaker than the entangled sampling techniques used in \cref{alg:learning_weyl_bell} and prior learning algorithms based on Bell sampling \cite{montanaro-bell-sampling,gross2021schur,lai2022learning,grewal2023improved}, 
computational difference sampling still reveals some information about the unsigned stabilizer group of a quantum state, as we will show in this section. 
We begin by proving some basic properties about $r_\psi$. 

\begin{fact}\label{fact:r-psi-convolution}
    Let $\ket\psi = \sum_{x \in \F_2^n} \alpha_x \ket{x}$ be an $n$-qubit quantum state, and let $a \in \F_2^n$. Then 
    \[
    r_\psi(a,0^n) = \sum_{x \in \F_2^n} \abs{\alpha_x}^2 \abs{\alpha_{x+a}^2}.
    \]
\end{fact}
\begin{proof}
It is a basic fact in probability theory that the sum of two independent random variables is the convolution of their individual distributions. 
By the Born rule, we observe $x \in \F_2^n$ with probability $\abs{\alpha_x}^2$ upon measuring $\ket\psi$ in the computational basis. Therefore, the probability of computational difference sampling $(a,0^n) \in \Xs$ is precisely 
\[
    r_\psi(a,0^n) = \sum_{x \in \F_2^n} \abs{\alpha_x}^2 \abs{\alpha_{x+a}^2},
\]
which is the convolution of the Born distribution with itself.
Equivalently, this is the probability of sampling two strings $x$ and $y$ whose difference is $a$.
\end{proof}

Next, we relate $r_\psi$ to $p_\psi$ (defined in \cref{subsec:weyl-expansion-and-bell-diff-sampling}).  
To do so, we need the following proposition, which relates sums of basis state projections to sums of Pauli operators. 

\begin{proposition}\label{prop:comp-diff-sampling-measurement}
For any $a \in \F_2^n$,
    \[\sum_{x \in \F_2^n} \ketbra{x}{x} \otimes \ketbra{x + a}{x+a} = \frac{1}{2^n}\sum_{x \in \F_2^n} (-1)^{a \cdot x} Z^x \otimes Z^x.
    \]
\end{proposition}
\begin{proof}
   \begin{align*}
       \sum_{x \in \F_2^n} \ketbra{x}{x} \otimes \ketbra{x+a}{x+a}
   &= \sum_{x \in \F_2^n} \ketbra{x}{x} \otimes X^a \ketbra{x}{x} X^a \\
   &= (I \otimes X^a) \left(\sum_{x \in \F_2^n} \ketbra{x}{x} \otimes \ketbra{x}{x}\right) \left(I \otimes X^a\right)  \\
   &=\frac{1}{2^n} (I \otimes X^a) \left(\sum_{x \in \F_2^n} Z^x \otimes Z^x\right) \left(I \otimes X^a\right) && \text{(\cref{prop:sum-over-stabilizer-basis})}  \\
   &= \frac{1}{2^n} \sum_{x \in \F_2^n} Z^x \otimes X^a Z^xX^a   \\
   &= \frac{1}{2^n} \sum_{x \in \F_2^n} (-1)^{a \cdot x} Z^x \otimes Z^x.&&\qedhere
   \end{align*} 
\end{proof}

The formal relation between $r_\psi$ and $p_\psi$ is stated below. Intuitively, this shows that $r_\psi$ is obtained from $p_\psi$ by summing over all possible shifts by a Pauli-$Z$ operator.

\begin{proposition}\label{prop:r-mass-vs-p-mass}
Let $\ket\psi$ be an $n$-qubit pure state, and let $a \in \F_2^n$. Then
    \[r_\psi(a,0^n) = \sum_{w \in \F_2^n} p_\psi(a, w) = \sum_{x \in a \times \F_2^n}p_\psi(x).\]
\end{proposition}
\begin{proof}
We can always write $\ket\psi = \sum_{x \in \F_2^n} \alpha_x \ket{x}$.
    \begin{align*}
        \sum_{x \in \F_2^n}p_\psi(a, x) 
        &= \frac{1}{2^n} \sum_{x \in \F_2^{n}} \braket{\psi|i^{a \cdot x}X^a Z^x|\psi}^2 \\
        &= \frac{1}{2^n} \sum_{x \in \F_2^{n}} (-1)^{a \cdot x} \braket{\psi|X^a Z^x|\psi}^2 \\
        &= \frac{1}{2^n} \sum_{x \in \F_2^{n}} (-1)^{a \cdot x}  \bra{\psi}^{\otimes 2} \left(X^a Z^x \otimes X^a Z^x\right)\ket{\psi}^{\otimes 2} \\
        &= \frac{1}{2^n} \sum_{x \in \F_2^{n}} (-1)^{a \cdot x}   \bra{\psi}^{\otimes 2}\left(X^a \otimes X^a\right) \left(Z^x \otimes Z^x\right)\ket{\psi}^{\otimes 2} \\
        &= \sum_{x \in \F_2^{n}}   \bra{\psi}^{\otimes 2}\left(X^a \otimes X^a\right) \left(\ketbra{x}{x} \otimes \ketbra{x+a}{x+a} \right)\ket{\psi}^{\otimes 2}&&\text{(\cref{prop:comp-diff-sampling-measurement})} \\
        &= \sum_{x \in \F_2^{n}}   \bra{\psi}^{\otimes 2}\left(\ketbra{x+a}{x} \otimes \ketbra{x}{x+a} \right)\ket{\psi}^{\otimes 2} \\
        &= \sum_{x \in \F_2^{n}} \abs{\alpha_x}^2\abs{\alpha_{x+a}}^2\\ 
        &= r_\psi(a,0^n).&&\text{(\cref{fact:r-psi-convolution})}\qedhere
    \end{align*}
\end{proof}

The duality property of $p_\psi$ with respect to the symplectic product, as captured by \cref{thm:p_q_duality}, extends to $r_\psi$ as a consequence of \cref{prop:r-mass-vs-p-mass}.

\begin{corollary}\label{cor:r-mass-vs-p-mass-duality}
    Given a subspace $H \subseteq \Xs$, then 
    \[
    \sum_{x \in H} r_\psi(x) = \frac{|H + \Zs|}{2^n}\sum_{x \in (H + \Zs)^\sympcomp} p_\psi(x).
    \]
\end{corollary}
\begin{proof}
\begin{align*}
    \sum_{x \in H} r_\psi(x) &= \sum_{x \in H + \Zs} p_\psi(x) && (\text{\cref{prop:r-mass-vs-p-mass}})\\
    &= \frac{|H + \Zs|}{2^n}\sum_{x \in (H + \Zs)^\sympcomp} p_\psi(x) && (\text{\cref{thm:p_q_duality}})\qedhere
\end{align*}
\end{proof}

Next, we relate computational difference samples of $\ket\psi$ to the Pauli-$Z$ strings that stabilize $\ket\psi$ modulo phase. In effect, \cref{cor:comp-diff-sampling} tells us that learning the support of $r_\psi$ allows us to learn the Pauli-$Z$ operators that stabilize $\ket{\psi}$ (up to a $\pm 1$ phase). 

\begin{corollary}\label{cor:comp-diff-sampling}
For $x \in \Xs$, $r_\psi(x) > 0$ implies $x \in (\weyl(\ket \psi) \cap \Zs)^\sympcomp$.
\end{corollary}
\begin{proof}
Set $H = (\weyl(\ket \psi) \cap \Zs)^\sympcomp \cap \Xs$ in \cref{cor:r-mass-vs-p-mass-duality}. Then, by \cref{lem:A-plus-B-sympcomp-identity}, 
\begin{align*}
    (H + \Zs)^\sympcomp &= H^\sympcomp \cap \Zs\\
    &= (\weyl(\ket \psi) \cap \Zs + \Xs) \cap \Zs\\
    &= \weyl(\ket \psi) \cap \Zs\\
    &\subseteq \weyl(\ket \psi).
\end{align*}
It follows that:
\begin{align*}
    \sum_{x \in H} r_\psi(x) &= \frac{|H + \Zs|}{2^n} \sum_{x \in (H + \Zs)^\sympcomp} p_\psi(x) && (\text{\cref{cor:r-mass-vs-p-mass-duality}})\\
    &= \frac{|H + \Zs|}{2^n} \cdot \frac{|(H + \Zs)^\sympcomp|}{2^n} && ((H + \Zs)^\sympcomp \subseteq \weyl(\ket \psi))\\
    &= 1 && (\text{\cref{fact:sympcomp}})
\end{align*}
We conclude that the probability of sampling something outside of $H$ is zero.
\end{proof}

\section{Learning the Unsigned Stabilizer Group Using Single-Copy Measurements}

In \cref{sec:learning_weyl_bell}, we gave an algorithm to learn a subspace $S$ satisfying the conditions of \cref{thm:learn-weyl-psi-reduction} using entangled measurements. 
Here, we present an algorithm that learns such a subspace $S$ using only single-copy measurements. We begin by stating the algorithm.

\begin{algorithm}[H]
\caption{Approximating $\weyl(\ket \psi)$ using single-copy measurements}\label{alg:learning_weyl_single}
\SetKwInOut{Promise}{Promise}
\KwInput{Black-box access to copies of $\ket\psi$ and $\eps, \delta \in (0,1)$}
\Promise{$\ket\psi$ has stabilizer dimension at least $n-t$}
\KwOutput{A subspace $S \subset \F_2^{2n}$ of dimension at least $n-t$ such that $\sum_{x \in S} p_\psi(x) \geq \left(1-\frac{\eps^2}{4}\right)\frac{\abs{S}}{2^n}$ with probability at least $1 - \delta$}

Set $\mclifford \gets 2\left(2^{t+1}+1\right)\left(n + \log(2/\delta)\right)$.

\For{$i \gets 0$ \KwTo $\mclifford$}{

Use the algorithm in \cref{prop:random-clifford-sample} to sample a random Clifford circuit $C_i$, and apply $C_i$ to $\frac{32n}{\eps^2}\left(n + \log(2 \mclifford/\delta)\right)$ copies of $\ket{\psi}$. Define $\ket{\psi_i} \coloneqq C_i \ket\psi$.
\label{step:one}

Computational difference sample $\ket{\psi_i}$ to draw $\frac{16n}{\eps^2}\left(n + \log(2 \mclifford/\delta)\right)$ samples from $r_{\psi_i}$.\label{step:two} 

Let $\hat{H}_i \subseteq \Xs$ be the subspace spanned by the computational difference samples. 
Use Gaussian elimination to compute generators of $\hat{H}_i$. 
\label{step:three}
}

Using \cref{lemma:compute-symplectic-complement}, compute $S \coloneqq \sum_{i=1}^{m_\mathsf{Clifford}} C_i^\dagger \left((\hat{H}_i + \Zs)^\sympcomp \right)$ as the sum of subspaces over $\F_2^{2n}$.\footnotemark\label{step:four}

\Return $S$
\end{algorithm}
\footnotetext{By cleverly using the regular orthogonal complement over $\F_2^n$, computing the symplectic complement of subspaces of the form $H + \calZ$ will offer a computational savings by a multiplicative factor of about $4$ over na\"ively using \cref{lemma:compute-symplectic-complement}. This is because, if we let $\overline{H}$ be the projection of $H$ onto $\F_2^n$ obtained by ignoring the last $n$ coordinates, and let $\overline{G}$ be the orthogonal complement of $\overline{H}$, then $(H + \calZ)^\sympcomp = \left(\overline{H} \times \F_2^n\right)^\sympcomp = 0^n \times \overline{G}$.
}
\vspace{\baselineskip}

Let us briefly explain the workings of the algorithm at an intuitive level.
The idea is to use computational difference sampling to learn about $\weyl(\ket \psi)$ via the support of $r_\psi$.
\Cref{cor:comp-diff-sampling} says that 
computational difference sampling allows us to learn the generators of $\weyl (\ket \psi)$ that are also contained in $\calZ$.
Of course, unless $\weyl(\ket \psi) \subseteq \calZ$, we will not learn a complete set of generators for $\weyl(\ket \psi)$. 
However, if we apply a random Clifford circuit $C$ to $\ket \psi$, then $\weyl(C \ket \psi) \cap \calZ$ will be nontrivial with some probability, and therefore reveal additional generators of $\weyl(\ket \psi)$. 
We repeat this process many times until we are confident that we have learned a complete set of generators for $\weyl(\ket\psi)$. 
Doing so requires more samples when $\weyl(\ket \psi)$ is small, because then $\weyl(C \ket \psi) \cap \calZ$ is more likely to be empty---this is why $\mclifford$ grows exponentially in $t$.

In practice, we will not be able to guarantee that we exactly learn the support of $r_\psi$.
This also means that we cannot guarantee that we learn $\weyl(\ket \psi)$ exactly.
Instead, we show that the output of \cref{alg:learning_weyl_single} satisfies the conditions in \cref{thm:learn-weyl-psi-reduction}.

\subsection{The Output Has Large Dimension}\label{ssec:large-dim}

The goal of this section is to show that the output of \cref{alg:learning_weyl_single} has dimension at least $n-t$. Since $\weyl(\ket \psi)$ naturally satisfies this property by assumption, if we can show that the output of \cref{alg:learning_weyl_single} is a superset of $\weyl(\ket \psi)$, then this output must also have large dimension.

We start by showing that the probability of learning a new element of $\weyl(\ket \psi)$ (should one exist) is not \emph{too} difficult when a uniformly random Clifford circuit is applied. In what follows, think of $A$ as $\weyl(\ket \psi)$, $T$ as the subspace of $A$ learned so far, and $B$ as the set of Pauli-$Z$ operators.

\begin{lemma}
\label{lem:sample-in-I-Z}
    Let $A$ be an isotropic subspace of dimension $n - k$, let $T \subset A$ be a proper subspace of $A$, let $B$ be a fixed Lagrangian subspace (e.g.\ $\Zs$), and let $C$ be a uniformly random Clifford circuit. Then
    \[
    \Pr[C(A \setminus T) \cap B \neq \emptyset] \ge \frac{1}{2^{k+1} + 1}.
    \]
\end{lemma}

\begin{proof}
    Without loss of generality, assume that $T$ has dimension $n - k - 1$, because the probability we want to bound can only be larger for smaller subspaces $T$. Write $A = \langle x \rangle + T$ for some generator $x \in T^{\sympcomp}$, so that $A \setminus T$ is equivalent to the coset $x + T$. Observe that
    \begin{align*}
        \Pr[C(A \setminus T) \cap B \neq \emptyset] &= \Pr[C(x + T) \cap B \neq \emptyset]\\
        &= \Pr[C(x) \in C(T) + B].
    \end{align*}
    We complete the proof by computing this quantity.
    We can think of sampling $C(T)$ and $C(x)$ by first choosing $C(T)$, and then sampling $C(x)$ conditioned on $C(T)$. Observe that in doing so, conditioned on a choice of $C(T)$, $C(x)$ is uniformly distributed in $C(T^\sympcomp \setminus T)$. As a result, it holds that
    \begin{align*}
        \Pr[C(x) \in C(T) + B \mid C(T)] &=  \frac{\abs{\left(C(T) + B\right) \cap C\left(T^\sympcomp \setminus T\right)}}{\abs{C\left(T^\sympcomp \setminus T\right)}}\\
        &= \frac{\abs{\left(C(T) + B\right) \cap C\left(T^\sympcomp\right)} - \abs{C(T)}}{\abs{C\left(T^\sympcomp\right)} - \abs{C(T)}}\\
        &= \frac{\abs{\left(C(T) + B\right) \cap C\left(T\right)^\sympcomp} - \abs{C(T)}}{\abs{C\left(T^\sympcomp\right)} - \abs{C(T)}}\\
        &= \frac{2^n - 2^{n-k-1}}{2^{n+k+1} - 2^{n-k-1}} && (\text{\cref{cor:A-plus-B-sympcomp-Lagrangian}})\\
        &= \frac{1}{2^{k+1} + 1}.
    \end{align*}
    Because this holds for every choice of $C(T)$, the lemma follows.
\end{proof}

Now that we have an upper bound on the difficulty of learning a new generator of $\weyl(\ket \psi)$, we can analyze the number of Clifford circuits we need in total to guarantee that the output of \cref{alg:learning_weyl_single} is a superset of $\weyl(\ket \psi)$. The analysis is very similar to \cref{lem:sampling}, though we are no longer interested in approximating the support of a distribution.

\begin{lemma}\label{lem:num-random-clifford-basis}
    Let $\ket \psi$ be a quantum state with stabilizer dimension $n-k$. Consider sampling $m$ random Clifford circuits $C_1, \ldots, C_m$, and let $S_i = \weyl(\ket \psi) \cap C_i^\dagger (\Zs)$. If
    \[
    m \ge 2\left(2^{k+1}+1\right)\left(n + \log(1/\delta)\right),
    \]
    then with probability at least $1 - \delta$, $\sum_{i=1}^m S_i = \weyl(\ket \psi)$.

\end{lemma}
\begin{proof}
Let $T_i = \sum_{j=1}^i S_j$, with the convention that we define $T_0$ to be the trivial subspace. Define the indicator random variable $X_i$ as
    
    \[
    X_i = \begin{cases}
    1 & S_i \not\subseteq  T_{i-1} \textrm{ or } T_{i-1} \supseteq \weyl(\ket \psi) \\
    0 & \textrm{otherwise}.
    \end{cases}
    \]

    Informally, $X_i = 1$ indicates a step at which the algorithm has made progress towards sampling a complete set of generators for $\weyl(\ket \psi)$ as part of $T_m$. It suffices to show that with probability $1 - \delta$, $\sum_{i=1}^{m} X_i \geq n$, as this guarantees that $T_{m} \supseteq \weyl(\ket \psi)$, because $\weyl(\ket \psi)$ has dimension at most $n$ by assumption.

    Let $c = 2^{k+1} + 1$.
    If $T_{i-1} \supseteq \weyl(\ket \psi)$ then $\Ex\left[X_i\right] = 1 \geq \frac{1}{c}$.
    Otherwise, $T_{i-1}$ is missing some member of $\weyl(\ket \psi)$, such that $T_{i-1} \cap \weyl(\ket \psi)$ has dimension at most $n-(k+1)$.
    By \cref{lem:sample-in-I-Z} with $A = \weyl(\ket \psi)$, $T = T_{i-1}$, and $B = \Zs$, 
    \begin{align*}
        \Ex\left[X_i \mid T_{i-1} \not\supseteq \weyl(\ket \psi)\right]
        &= 
        \Pr\left[A \cap C^\dagger(B) \not\subseteq T
        \mid T \subset A
        \right]\\
        &= 
        \Pr\left[(A \setminus T) \cap C^\dagger(B) \neq \emptyset
        \mid T \subset A
        \right]
        \\
        &= 
        \Pr\left[C(A \setminus T) \cap B \neq \emptyset
        \mid T \subset A
        \right]\\
        &\geq \frac{1}{1+2^{k+1}}\\
        &= \frac{1}{c}.
    \end{align*}
    So, $\Ex\left[X_i \right] \geq \frac{1}{c}$ regardless of what $T_i$ is.
    
    Let $\gamma = 1 - \frac{n}{c m}$.
    Then, by the multiplicative Chernoff bound (\cref{fact:chernoff}), we have
    \begin{align}
    \Pr\left[\sum_{i=1}^{m} X_i < n \right]
    &= \Pr\left[\sum_{i=1}^{m} X_i < (1 - \gamma) c m \right]\nonumber\\
    &\leq \exp\left(- \gamma^2 \frac{cm}{2}\right)\nonumber\\
    &= \exp\left(- \left(1 - \frac{2n}{cm} + \frac{n^2}{c^2m^2}\right) \frac{cm}{2}\right)\nonumber\\
    &\le \exp\left(- \left(1 - \frac{2 n}{cm}\right) \frac{cm}{2}\right)\nonumber\\
    &= \exp\left(n - \frac{cm}{2}\right)\label{eq:chernoff_estimate}
    \end{align}
    Hence, choosing
    \[
    m \ge \frac{2}{c}\left(n + \log\frac{1}{\delta}\right) = 2\left(2^{k+1}+1\right)\left(n + \log\frac{1}{\delta}\right)
    \]
    suffices to guarantee that \cref{eq:chernoff_estimate} is at most $\delta$.\qedhere
\end{proof}

\subsection{The Output Has Large \texorpdfstring{$p_\psi$}{p}-Mass}\label{ssec:large-p-mass}
In this subsection, we argue that the subspace output by \cref{alg:learning_weyl_single} satisfies \Cref{item:approx-weyl} of \cref{thm:learn-weyl-psi-reduction} (i.e., it accounts for a near-maximal $p_\psi$-mass).
In \cref{step:four}, we take the span of subspaces, with the hope that each one learns the support of $r_{\psi_i}$ exactly.
This would in turn allow us to \emph{only} learn elements of $\weyl(\ket \psi)$ via \cref{cor:comp-diff-sampling}.
However, it is possible that we will fail to do so in practice. For example, if $\ket{\psi_i}$ is close to (but not exactly) a computational basis state, then the support of $r_{\psi_i}$ is nontrivial, even though computational difference sampling will almost always sample the identity Weyl operator.

Nevertheless, we argue that even if we do not learn the support of $r_{\psi_i}$ exactly, we can still succeed.
In particular, \cref{cor:r-mass-vs-p-mass-duality} also tells us that learning a large fraction of the $r_{\psi_i}$-mass corresponds to learning a subspace that has near maximal $p_\psi$-mass (i.e., even though we pick up elements outside of $\weyl(\ket \psi)$, they still have large $p_\psi$ value on average).
Therefore, for \cref{step:four} to work, we need to show that the sums of subspaces of $\F_2^{2n}$ with near-maximal $p_\psi$-mass form a (potentially) larger subspace that is still near-maximal. The next lemma establishes this.

\begin{lemma}\label{cor:subspace-addition-p-mass}
    Let $S_1, \dots, S_k$ be subspaces of $\F_2^{2n}$ such that $\sum_{x \in S_i} p_\psi(x) \geq (1-\eps_i)\frac{\abs{S_i}}{2^n}$ for all $S_i$.
    Then $S = \sum_{i=1}^k S_i$ satisfies
    \[
        \sum_{x \in S} p_\psi(x) \geq \left(1-\sum_{i=1}^k \eps_i\right)\frac{\abs{S}}{2^n}.
    \]
\end{lemma}
\begin{proof}
    We proceed via induction on $k$. The base case $k=1$ follows by assumption. Now, take $T = \sum_{i=1}^{k-1}S_i$ such that
    \[
        \sum_{x \in T} p_\psi(x) \geq \left(1 - \sum_{i=1}^{k-1} \eps_i\right) \frac{\abs{T}}{2^n}
    \]
    by the inductive hypothesis. Then,
    \begin{align*}
        \sum_{x \in S} p_\psi(x) &= \frac{\abs{T+S_k}}{2^n} \sum_{x \in (T+S_k)^\sympcomp} p_\psi(x) && (\text{\cref{thm:p_q_duality}})\\
        &= \frac{\abs{T+S_k}}{2^n} \sum_{x \in T^\sympcomp \cap S_k^\sympcomp} p_\psi(x) && (\text{\cref{lem:A-plus-B-sympcomp-identity}})\\
        &= \frac{\abs{T+S_k}}{2^n} \left( \sum_{x \in T^\sympcomp} p_\psi(x) + \sum_{x \in S_k^\sympcomp}p_\psi(x) -\!\! \sum_{x \in T^\sympcomp \cup S_k^\sympcomp}p_\psi(x)\right) && (\text{Inclusion-Exclusion})\\
        &\geq \frac{\abs{T+S_k}}{2^n} \left( \frac{2^n}{\abs{S}}\sum_{x \in T} p_\psi(x) + \frac{2^n}{\abs{S_k}}\sum_{x \in S_k}p_\psi(x) - 1\right)\\
        &=  \left(1- \eps_k -\sum_{i=1}^{k-1} \eps_i\right)\frac{\abs{T+S_k}}{2^n} = \left(1 -\sum_{i=1}^{k} \eps_i\right)\frac{\abs{S}}{2^n}.&&\qedhere
    \end{align*}
\end{proof}

Let $m$ be the number of Clifford circuits needed from \cref{ssec:large-dim}.
For \cref{alg:learning_weyl_single} to fulfill \Cref{item:approx-weyl} of \cref{thm:learn-weyl-psi-reduction}, we need each of the $\hat{H}_i$ from \cref{step:three} to only capture $x \in \F_2^{2n}$ such that the average squared expectation, i.e.,
\[
\frac{2^n}{\abs{\hat{H}_i}}\sum_{x \in \hat{H}_i}p_\psi(x) = \E_{x \sim \hat{H}_i} \left[\braket{\psi|W_x|\psi}^2\right],
\]
is very close to $1$.
Using \cref{cor:subspace-addition-p-mass}, one might na\"ively expect to require an average squared expectation of $1-\frac{\eps}{m}$ for each $\hat{H}_i$.
However, we can provide a slight optimization by noting that, since all of the subspaces lie in $\F_2^{2n}$, we only need to ever consider at most $2n$ subspaces when \cref{cor:subspace-addition-p-mass} is applied.
As a result, we only require that the average squared expectation is at least $1-\frac{\eps}{2n}$ for each $\hat{H}_i$.

\begin{lemma}
\label{lem:sum-subspace-over-2n-dim}
Let $S_1, \dots, S_m$ denote $m$ (potentially non-unique) subspaces of $\F_2^{2n}$ and 
let $S = \sum_{i=1}^m S_i$. Suppose that for some $\eps > 0$ and all $S_i$,
\[
\sum_{x \in S_i} p_\psi(x) \geq \left(1 - \frac{\eps}{2n}\right)\frac{\abs{S_i}}{2^n}.
\]
Then, 
\[
\sum_{x \in S} p_\psi(x)  \geq \left(1 - \eps\right)\frac{\abs{S}}{2^n}.
\]
\end{lemma}
\begin{proof}
    If $m \leq 2n$ then we can directly use \cref{cor:subspace-addition-p-mass} (with $k \coloneqq m$) and we are done.
    Otherwise, also by \cref{cor:subspace-addition-p-mass} (with $k \coloneqq 2n$), for every subset $I \subseteq [m]$ of $2n$ indices (i.e., $\abs{I} = 2n$), \[ \sum_{x \in S_I} p_\psi(x) \geq \left(1 - \sum_{i \in I}\frac{\eps}{2n}\right)\frac{\abs{S_I}}{2^n} = \left(1 - \eps\right)\frac{\abs{S_I}}{2^n},\]
    where $S_I \coloneqq  \sum_{i \in I} S_i$.
    Because all of these subspaces lie in $\F_2^{2n}$, we note that $S = S_I$ for some $I \subseteq [m]$ of cardinality $2n$.\footnote{A more careful analysis involving \cref{lem:p-mass-isotropic} reveals that $S$ must be isotropic for $\eps < \frac{1}{4}$, and therefore $\dim(S) \leq n$. Thus $I$ can be taken to have cardinality $n$ rather than $2n$. Since $\eps$ is always less than $\frac{1}{4}$ in \cref{alg:learning_weyl_single}, this halves the number of samples needed in \cref{step:two} and, therefore, \cref{alg:learning_weyl_single} as a whole. For the sake of clarity, we do not go through this analysis.}
    Therefore, it must also be the case that
    \[
        \sum_{x \in S} p_\psi(x)  \geq \left(1 - \eps\right)\frac{\abs{S}}{2^n}. \qedhere
    \]
\end{proof}

Using \cref{cor:r-mass-vs-p-mass-duality} and \cref{fact:clifford-permutes-p-mass}, we can translate \cref{lem:sum-subspace-over-2n-dim} into something more closely related to analyzing \cref{alg:learning_weyl_single}.

\begin{proposition}\label{prop:large-p-mass}
Let $C_1, \dots, C_m$ denote $m$ (potentially non-unique) $n$-qubit Clifford circuits and 
let $\ket{\psi_i} \coloneqq C_i \ket \psi$ for some fixed $n$-qubit quantum state $\ket \psi$. Given subspaces $H_1,\ldots,H_m \subseteq \Xs$, suppose that for some $\eps > 0$ and for all $H_i$,
\[
\sum_{x \in H_i} r_{\psi_i}(x) \geq 1 - \frac{\eps}{2n}.
\]
Define $S \coloneqq \sum_{i=1}^{m} C_i^\dagger \left((H_i + \Zs)^\sympcomp \right)$. Then
\[
\sum_{x \in S} p_\psi(x)  \geq \left(1 - \eps\right)\frac{\abs{S}}{2^n}.
\]
\end{proposition}
\begin{proof}
For convenience, define $S_i = C_i^\dagger \left((H_i + \Zs)^\sympcomp \right)$. For each $S_i$, we have
\begin{align*}
    \sum_{x \in S_i} p_{\psi}(x) &=
    \sum_{x \in C^\dagger\left((H_i + \Zs)^\sympcomp\right)} p_{\psi}(x)\\
    &= \sum_{x \in (H_i + \Zs)^\sympcomp} p_{\psi_i}(x) && (\text{\cref{fact:clifford-permutes-p-mass}})\\
    &= \frac{2^n}{\abs{H_i + \Zs}}\sum_{x \in H_i}r_{\psi_i}(x) && (\text{\cref{cor:r-mass-vs-p-mass-duality}})\\ 
    &= \frac{\abs{(H_i + \Zs)^\sympcomp}}{2^n} \sum_{x \in H_i}r_{\psi_i}(x) && (\text{\cref{fact:sympcomp}})\\
    &\geq \left(1 - \frac{\eps}{2n}\right)\frac{|(H_i + \Zs)^\sympcomp|}{2^n}\\
    &= \left(1 - \frac{\eps}{2n}\right)\frac{|S_i|}{2^n},
\end{align*}
    using in the last line the fact that the Clifford action on subspaces preserves dimension.
    Finally, \cref{lem:sum-subspace-over-2n-dim} tells us that
    \[
        \sum_{x \in S} p_\psi(x) \geq \left(1 - \eps\right)\frac{\abs{S}}{2^n}
    \]
    as desired.
\end{proof}

\subsection{Putting It All Together}

We now have all of the tools to show that \cref{alg:learning_weyl_single} succeeds at learning a subspace that satisfies \Cref{item:dimension,item:approx-weyl} of \cref{thm:learn-weyl-psi-reduction}.

\begin{theorem}\label{thm:single-copy-main}
    Let $\ket{\psi}$ be an $n$-qubit quantum state with stabilizer dimension at least $n-t$.
    Given copies of $\ket{\psi}$ as input, \cref{alg:learning_weyl_single} succeeds at outputting a subspace $S$ that satisfies the conditions of \cref{thm:learn-weyl-psi-reduction} with probability at least $1-\delta$.
    The algorithm uses
    \[
        O\left(\frac{n\left(n + \log(1/\delta)\right)^2 2^t}{\eps^2}\right)
    \]
    samples and
    \[
        O\left(\frac{n^3\left(n + \log(1/\delta)\right)^2 2^t}{\eps^2}\right)
    \]
    time.
\end{theorem}
\begin{proof}
    We will first show that \Cref{item:dimension} is met.
    By setting $\mclifford = 2\left(2^{t+1}+1\right)\left(n + \log(2/\delta)\right)$, \cref{lem:num-random-clifford-basis} tells us that the subspaces $S_i = \weyl(\ket \psi) \cap C_i^\dagger(\Zs)$ span $\weyl(\ket \psi)$, except with with failure probability at most $\delta/2$. Assume that we succeed. Then \Cref{cor:comp-diff-sampling} shows that we are guaranteed to only sample $y \sim r_{\psi_i}$ such that $y \in H_i$, where
    \[
    H_i \coloneqq \left(\weyl(\ket{\psi_i}) \cap \Zs\right)^\sympcomp \cap \Xs
    \]
    Hence, $\hat{H_i} \subseteq H_i$, and therefore
    \begin{align*}
        \left(\hat{H}_i + \Zs \right)^\sympcomp &\supseteq \left(H_i + \Zs \right)^\sympcomp\\
        &= H_i^\sympcomp \cap \calZ  && (\mathrm{\cref{lem:A-plus-B-sympcomp-identity}})\\
        &=  \left(\left(\weyl(\ket{\psi_i}) \cap \Zs\right) + \calX\right) \cap \calZ && (\mathrm{\cref{lem:A-plus-B-sympcomp-identity}})\\
        &= (\calX \cap \calZ) + \left(\weyl(\ket \psi_i) \cap \Zs\right) && (\mathrm{\cref{lem:A-plus-B-sympcomp-identity-2}})\\
        &= \weyl(\ket{\psi_i}) \cap \Zs\\
        &= C(S_i).
    \end{align*}
    By the end of the $\mclifford$ random Clifford circuits, it follows that
    \[
    S = \sum_{i=1}^{\mclifford} C_i^\dagger \left((\hat{H}_i + \Zs)^\sympcomp \right) \supseteq \sum_{i=1}^{\mclifford} S_i = \weyl(\ket \psi)
    \]
    and therefore $\dim(S) \geq \dim\left(\weyl(\ket \psi)\right) \geq n-t$.
    
    Moving on to \Cref{item:approx-weyl}, \cref{lem:sampling} shows that if we draw $\frac{16n}{\eps^2}\left(n + \log(2 \mclifford/\delta)\right)$ computational difference samples per $\ket{\psi_i}$ then
    \begin{align}\label{eq:inner-loop-condition}
        \sum_{x \in \hat{H}_i} r_{\psi_i}(x) \geq 1 - \frac{\eps^2}{8n}
    \end{align}
    for all $\hat{H}_i$ with failure probability at most $\delta/2$.
    Again, assume that we succeed.
    By \cref{prop:large-p-mass}, it follows that
    \[
        \sum_{x \in S} p_\psi(x) \geq \left(1 - \eps^2/4\right)\frac{\abs{S}}{2^n},
    \]
    thus satisfying \Cref{item:approx-weyl}.
    
    By the union bound, this means we satisfy \Cref{item:dimension,item:approx-weyl} with probability at least $1-\delta$.

    We next turn to the sample complexity analysis.
    For simplicity, let $\mcomp$ be the number of computational difference samples taken in each part of the inner loop.
    Due to the fact that $\log\left(\mclifford\right) = O\left(t + \log n + \frac{\log(2/\delta)}{n}\right)$, we find that
    \[\mcomp \coloneqq \frac{16n}{\eps^2}\left(n + \log(2\mclifford/\delta)\right) 
    = O\left(\frac{n\left(n + \log(1/\delta)\right)}{\eps^2}\right),
    \]
    which will greatly simplify our asymptotic analysis.
    Since a computational difference sample involves $2$ copies of $\ket \psi$, the total number of copies of $\ket \psi$ needed is
    \[
        2 \mclifford\mcomp= O\left(\frac{n\left(n + \log(1/\delta)\right)^2 2^t}{\eps^2}\right)
    \]

    To analyze the runtime, we will first analyze the inner loop of \cref{alg:learning_weyl_single}.
    By \cref{prop:random-clifford-sample}, we require $O(n^2)$ time to sample a random Clifford circuit in the first part of \cref{step:one}.
    Applying $C_i$ to $2 \mcomp$ copies of $\ket \psi$ in the second half \cref{step:one} requires $O(\mcomp n^2)$ time.
    In \cref{step:two}, we require $O(\mcomp n)$ time to make $O(\mcomp)$ many $n$-qubit measurements, then apply the difference to pairs of measurements.
    Finally, it costs $O(\mcomp n^2)$ time to run \cref{lemma:compute-symplectic-complement} in \cref{step:three}.
    Thus, the inner loop requires time $O(\mcomp n^2)$, which becomes
    \[
        O(\mclifford \mcomp n^2) = O\left(\frac{n^3\left(n + \log(1/\delta)\right)^2 2^t}{\eps^2}\right)
    \]
    over all iterations of the inner loop.
    This will turn out to be the dominating term.
    
    For completeness, we will analyze \cref{step:four}.
    Given the generators of $\hat{H}_i$, it takes $O(n^3)$ time to compute $\left(\hat{H}_i + \calZ\right)^\sympcomp$ using \cref{lemma:compute-symplectic-complement}. By the size of the Clifford circuit from \cref{prop:random-clifford-sample}, we also need $O(n^3)$ time to apply $C_i^\dagger$,\footnote{Since we are now classically simulating the action of $C_i^\dagger$ on a subspace of $\F_2^{2n}$ as part of our classical post-processing, it will require an extra $O(n)$ time per gate to apply $C_i^\dagger$.} for a total time of $O\left(\mclifford \cdot n^3\right)$ to compute generators of each $C_i^\dagger\left(\hat{H}_i + \calZ\right)^\sympcomp$.
    All that is left is to compute $\sum C_i^\dagger\left(\hat{H}_i + \calZ\right)^\sympcomp$.
    Since each $C_i^\dagger\left(\hat{H}_i + \calZ\right)^\sympcomp$ can be described by at most $2n$ generators, the generators of $S$ can be computed by running Gaussian elimination on a $2n\mclifford \times 2n$ matrix over $\F_2$.
    This also takes $O(\mclifford \cdot n^3) = O\left(n^3 \cdot 2^t \cdot \left(n + \log(1/\delta)\right)\right)$ time. \qedhere
\end{proof}

Combining with \cref{alg:reduction} gives us the following.

\begin{corollary}\label{cor:single-copy-main}
    Let $\ket{\psi}$ be an $n$-qubit quantum state with stabilizer dimension at least $n-t$.
    Given copies of $\ket{\psi}$ as input, the combination of \cref{alg:learning_weyl_single,alg:reduction} output a classical description of $\ket{\hat{\psi}}$ such that $\tracedistance{\ket{\psi}, \ket{\hat{\psi}}} \leq \eps$ with probability at least $1-\delta$.
    The algorithm uses
    \[
    O\left(\frac{n\left(n + \log(1/\delta)\right)^2 2^t}{\eps^2} + N_{t, \frac{\eps}{2}, \frac{\delta}{6}}\right)
    \]
    single-copy samples and
    \[
    O\left(\frac{n^3\left(n + \log(1/\delta)\right)^2 2^t}{\eps^2} + n^2 \cdot N_{t, \frac{\eps}{2}, \frac{\delta}{6}} \right) + M_{t, \frac{\eps}{2}, \frac{\delta}{6}},
    \]
    time.
\end{corollary}
\begin{proof}
    By \cref{thm:learn-weyl-psi-reduction,thm:single-copy-main}, each algorithm succeeds with probability at least $1-\delta/2$. By the union bound, the total failure probability is at most $\delta$.
\end{proof}

\section{Open Problems}

Given $\ket{\psi}$, a state that can be prepared with Clifford gates and $O(\log n )$ non-Clifford gates, our tomography algorithms efficiently learn a classical description of $\ket{\psi}$. However, our algorithms are not \textit{proper} learners---in other words, the circuit that the algorithms output to approximate the state does not necessarily decompose into few non-Clifford gates.
\begin{question}\label{question:find-circuit}
Given copies of an $n$-qubit state $\ket\psi$ that is the output of a Clifford circuit and $O(\log n)$ non-Clifford gates, is it possible to efficiently construct a circuit $U$ comprised of Clifford and $O(\log n)$ non-Clifford gates such that $U \ket{0^n}$ is $\eps$-close to $\ket\psi$?
\end{question}

We note that an efficient proper learning algorithm for stabilizer states is known (namely, just run \cref{alg:reduction} with either \cref{alg:learning_weyl_bell} or \cref{alg:learning_weyl_single} and $t=0$).  
We also note that this task becomes trivial if polynomially many non-Clifford gates are allowed in $U$ because one can use \cref{alg:reduction} to produce $C$, $\ket{x}$, and $\ket{\varphi}$, and then construct a circuit with at most $2^{O(\log n)} = \poly(n)$ general gates that outputs $\ket{\varphi}$ \cite{SBM06-synthesis}.

As a subroutine, \cref{alg:reduction} uses pure state tomography to recover a classical description of a pure state, and therefore the copy and time complexities of our algorithm can be improved if faster pure state tomography algorithms are developed. 
Developing a pure state tomography algorithm that achieves the optimal $O(2^n/\eps^2)$ copy and time complexities is an interesting and important direction for future work. 

Observe that \cref{alg:learning_weyl_bell} learns an approximation to $\weyl(\ket \psi)$ using only $O(n)$ two-copy measurements, independent of the size of $\weyl(\ket \psi)$. By contrast, \Cref{alg:learning_weyl_single} requires $\poly(n) \exp(t)$ samples to do the same when $\weyl(\ket \psi)$ has dimension $n - t$. This raises the natural question of whether an exponential dependence on $t$ is necessary using unentangled measurements:

\begin{question}
\label{q:single_copy_separation}
    Are $\exp(t)$ many single-copy measurements necessary to learn a subspace $S$ that satisfies the conditions of \cref{thm:learn-weyl-psi-reduction}?
\end{question}
Note that analogous exponential separations between entangled and single-copy measurements have been shown in other learning contexts \cite{Huang2021information,Chen2022exponential}. Moreover, these separations hold for closely-related learning tasks that may be viewed as identifying the large coefficients of the Weyl expansion (\cref{def:weyl_expansion}), but for mixed states. For this reason, we are optimistic that the techniques of \cite{Chen2022exponential} could be adapted to answer \cref{q:single_copy_separation} in the affirmative.

Finally, stabilizer dimension is a rather rigid notion, and the property of high stabilizer dimension excludes states that could be considered \say{near-stabilizer.} Could one extend the techniques in this paper and \cite{grewal2023improved} to states of high \emph{approximate} stabilizer dimension: that is, states which have high expectation with a large subgroup of Pauli operators?

\section*{Acknowledgements}
We thank Scott Aaronson, Richard Kueng, Marcel Hinsche, Nick Hunter-Jones, Luowen Qian, and John Wright for helpful conversations.

SG, VI, DL are supported by Scott Aaronson's Vannevar Bush Fellowship from the US Department of Defense, NSF QLCI Award OMA-2016245, a Simons Investigator Award, and the Simons ``It from Qubit'' collaboration. WK is supported by an NDSEG Fellowship and acknowledges support from the U.S.\ Department of Energy, Office of Science, National
Quantum Information Science Research Centers, Quantum Systems Accelerator. DL is also supported by NSF award FET-2243659. VI is also supported by the U.S. Department of Energy, Office of Science, Office of Advanced Scientific Computing Research, Accelerated Research in Quantum Computing and an NSF Graduate Research Fellowship.

This work was done in part while SG, VI, and DL were visiting the Simons Institute for the Theory of Computing.

This article has been authored by an employee of National Technology \& Engineering Solutions of Sandia, LLC under Contract No. DE-NA0003525 with the U.S. Department of Energy (DOE). The employee owns all right, title and interest in and to the article and is solely responsible for its contents. The United States Government retains and the publisher, by accepting the article for publication, acknowledges that the United States Government retains a non-exclusive, paid-up, irrevocable, world-wide license to publish or reproduce the published form of this article or allow others to do so, for United States Government purposes. The DOE will provide public access to these results of federally sponsored research in accordance with the DOE Public Access Plan \href{https://www.energy.gov/downloads/doe-public-access-plan}{https://www.energy.gov/downloads/doe-public-access-plan}.

\bibliographystyle{alphaurl}
\bibliography{refs}

\appendix

\section{A Faster Algorithm From One Round of Adaptivity}
\label{sec:appendix}
Both \cref{alg:learning_weyl_bell,alg:learning_weyl_single} are non-adaptive, in the sense that the measurements performed on $\ket{\psi}$ are independent of the outcomes of previous measurements, and so they can be measured in parallel.
In this appendix, we explain how to improve both the sample and time complexities of \cref{alg:learning_weyl_single}---the algorithm that only uses single-copy measurements---by using one round of adaptivity.\footnote{A similar modification can be done to \cref{alg:learning_weyl_bell}. However, the modification will not yield any asymptotic improvements on the sample and time complexities, like it will for our modification to \cref{alg:learning_weyl_single}.}  
Specifically, our modified algorithm will have two phases, where the second phase depends on the measurement outcomes of the first phase, while continuing to only use single-copy measurements. We credit the concurrent work by Chia, Lai, and Lin \cite{Chia2023} for the idea behind this modification.

The main idea of the algorithm is reasonably simple: suppose we knew an isotropic subspace $A$ that contains $\weyl(\ket \psi)$. Then by \cref{lem:clifford-mapping-algorithm}, we could find a Clifford circuit $C$ such that $\weyl(C\ket \psi) \subseteq \Zs$, by mapping $A$ to a subspace of $\Zs$. But notice that if all of the stabilizers of $C\ket \psi$ are contained in $\Zs$, then we can approximately learn $\weyl(C\ket \psi) = C\left(\weyl (\ket \psi)\right)$ from computational difference samples of $C\ket \psi$, as a consequence of \cref{cor:r-mass-vs-p-mass-duality}.

\textit{A priori}, it is not clear what this approach buys us, because learning an isotropic subspace $A$ that contains $\weyl(\ket \psi)$ could be roughly as hard as learning a subspace $S$ that satisfies the conditions of \cref{thm:learn-weyl-psi-reduction}. As we shall see, the advantage is that in \cref{step:one} of \cref{alg:learning_weyl_single}, we can get away with taking a number of computational difference samples that scales linearly in $n$, rather than quadratically in $n/\eps$.

We are now ready to state the algorithm.

\begin{algorithm}[H]
\caption{Approximating $\weyl(\ket \psi)$ using adaptive single-copy measurements}\label{alg:learning_weyl_single_adaptive}
\SetKwInOut{Promise}{Promise}
\KwInput{Black-box access to copies of $\ket\psi$ and $\eps, \delta \in (0,1)$}
\Promise{$\ket\psi$ has stabilizer dimension at least $n-t$}
\KwOutput{A subspace $S \subset \F_2^{2n}$ of dimension at least $n-t$ such that $\sum_{x \in S} p_\psi(x) \geq \left(1-\frac{\eps^2}{4}\right)\frac{\abs{S}}{2^n}$ with probability at least $1 - \delta$}

Set $\mclifford \gets 2\left(2^{t+1}+1\right)\left(n + \log(3/\delta)\right)$.
\label{step:one-2}

\For{$i \gets 0$ \KwTo $\mclifford$}{
\label{step:two-2}

Use the algorithm in \cref{prop:random-clifford-sample} to sample a random Clifford circuit $C_i$, and apply $C_i$ to $16\left(n + \log(3 \mclifford/\delta)\right)+2$ copies of $\ket{\psi}$. Define $\ket{\psi_i} \coloneqq C_i \ket\psi$.
\label{step:three-2}

Computational difference sample $\ket{\psi_i}$ to draw $8\left(n + \log(3 \mclifford/\delta)\right)+1$ samples from $r_{\psi_i}$.\label{step:four-2} 

Let $\hat{H}_i \subseteq \Xs$ be the subspace spanned by the computational difference samples. 
Use Gaussian elimination to compute generators of $\hat{H}_i$. 
\label{step:five-2}
}

Using \cref{lemma:compute-symplectic-complement}, compute $\hat{A} \coloneqq \sum_{i=1}^{\mclifford} C_i^\dagger \left((\hat{H}_i + \Zs)^\sympcomp \right)$.\footnotemark
\label{step:six-2}

Apply \cref{lem:clifford-mapping-algorithm} to find the Clifford circuit $C_{\Zs}$  that maps $\hat{A}$ to a subspace of $\calZ$.\label{step:seven-2}

Obtain $\frac{16\left(n + \log(3/\delta)\right)}{\eps^2}$ copies of $\ket{\psi^\prime} \coloneqq C_{\Zs} \ket \psi$.
\label{step:eight-2}

Computational difference sample all copies of $\ket{\psi^\prime}$ to obtain $\frac{8\left(n + \log(3/\delta)\right)}{\eps^2}$ samples from $r_{\psi^\prime}$.
\label{step:nine-2}

Let $\hat{H} \subseteq \calX$ be the subspace spanned by the computational difference samples. Use \cref{lemma:compute-symplectic-complement} to compute generators of $\left(\hat{H} + \calZ\right)^\sympcomp$.
\label{step:ten-2}

\Return $S \coloneqq C_{\Zs}^\dagger\left(\left(\hat{H} + \calZ\right)^\sympcomp\right)$.
\label{step:eleven-2}
\end{algorithm}
\footnotetext{Again, by smartly using the regular orthogonal complement over $\F_2^n$, computing the symplectic complement of subspaces of the form $H + \calZ$ will offer a computational savings by a multiplicative factor of about $4$ over na\"ively using \cref{lemma:compute-symplectic-complement}.}

Observe that \Cref{step:one-2,step:two-2,step:three-2,step:four-2,step:five-2,step:six-2} are the same as in \cref{alg:learning_weyl_single}, except that the number of computational difference samples taken in \cref{step:four-2} is smaller. 
The algorithm is adaptive only because the Clifford circuit obtained in \cref{step:seven-2} depends on earlier measurement outcomes. 

We now work to show the correctness of \cref{alg:learning_weyl_single_adaptive}. We first require a generalization of \cref{lem:p-mass-isotropic}. This ensures that $\hat{A}$ in \cref{step:six-2} will be isotropic, even when \cref{item:approx-weyl} is not met.
\begin{lemma}
\label{lem:p-mass-isotropic-2}
    Let $S$ be a subspace of $\F_2^{2n}$ such that
    \[
    \sum_{x \in S}p_{\psi}(x) > \frac{3}{4}\frac{\abs{S}}{2^n}.
    \]
    Then $S \subseteq \langle M \rangle$, where $M \coloneqq \{x \in \F_2^{2n} : 2^n p_\psi(x) > 1/2\}$.
\end{lemma}
\begin{proof}
    In the proof of \cref{lem:p-mass-isotropic}, we showed that $S = \langle A \rangle$, where $A \coloneqq \{x \in S: 2^n p_\psi(x) > 1/2\}$.
    Since $A \subseteq M$, we clearly have $S = \langle A \rangle \subseteq \langle M \rangle$, which proves the lemma.
\end{proof}

\begin{corollary}
    \label{cor:r-mass-isotropic-2}
    Let $H$ be a subspace of $\calX$ such that
    \[
    \sum_{x \in H}r_{\psi}(x) > \frac{3}{4}.
    \]
    Then $\left(H + \calZ\right)^\sympcomp \subseteq \langle M \rangle$, where $M \coloneqq \{x \in \F_2^{2n} : 2^n p_\psi(x) > \frac{1}{2}\}$.
\end{corollary}
\begin{proof}
    By \cref{cor:r-mass-vs-p-mass-duality} and \cref{fact:sympcomp},
    \[\sum_{x \in (H + \calZ)^\sympcomp} p_\psi(x) = \frac{2^n}{\abs{H + \calZ}}\sum_{x \in H} r_\psi(x) > \frac{3}{4}\frac{\abs{\left(H + \calZ\right)^\sympcomp}}{2^n}.\]
    \cref{lem:p-mass-isotropic-2} then shows that $\left(H + \calZ\right)^\sympcomp \subseteq \langle M \rangle$.
\end{proof}

We now have all of the tools to show that \cref{alg:learning_weyl_single_adaptive} succeeds at learning a subspace that satisfies \Cref{item:dimension,item:approx-weyl} of \cref{thm:learn-weyl-psi-reduction}.

\begin{theorem}\label{thm:single-copy-main-adaptive}
    Let $\ket{\psi}$ be an $n$-qubit quantum state with stabilizer dimension at least $n-t$.
    Given copies of $\ket{\psi}$ as input, \cref{alg:learning_weyl_single_adaptive} succeeds at outputting a subspace $S$ that satisfies all conditions of \cref{thm:learn-weyl-psi-reduction} with probability at least $1-\delta$.
    The algorithm uses
    \[
        O\left(\left(n + \log(1/\delta)\right)^2 2^t + \frac{n + \log(1/\delta)}{\eps^2}\right)
    \]
    samples and
    \[
        O\left(n^2\left(n + \log(1/\delta)\right)^2 2^t + \frac{n^2(n+\log(1/\delta))}{\eps^2}\right)
    \]
    time.
\end{theorem}

\begin{proof}
    The algorithm runs in two stages: \Cref{step:one-2,step:two-2,step:three-2,step:four-2,step:five-2,step:six-2} constitute the first non-adaptive stage, while \Cref{step:seven-2,step:eight-2,step:nine-2,step:ten-2,step:eleven-2} make up the second stage based on $\hat{A}$.
    
    We will first show that at the end of stage one, $\weyl(\ket \psi) \subseteq \hat{A}$ and $\hat{A}$ is isotropic, with high probability.
    By setting $\mclifford = 2\left(2^{t+1}+1\right)\left(n + \log(3/\delta)\right)$, \cref{lem:num-random-clifford-basis} tells us that the subspaces $S_i = \weyl(\ket \psi) \cap C_i^\dagger(\Zs)$ span $\weyl(\ket \psi)$, except with with failure probability at most $\delta/3$. Assume that we succeed.
    Then \Cref{cor:comp-diff-sampling} shows that we are guaranteed to only sample $y \sim r_{\psi_i}$ such that $y \in H_i$, where
    \[
    H_i \coloneqq \left(\weyl(\ket{\psi_i}) \cap \Zs\right)^\sympcomp \cap \Xs
    \]
    Hence, $\hat{H_i} \subseteq H_i$, and therefore
    \begin{align*}
        \left(\hat{H}_i + \Zs \right)^\sympcomp &\supseteq \left(H_i + \Zs \right)^\sympcomp\\
        &= H_i^\sympcomp \cap \calZ  && (\mathrm{\cref{lem:A-plus-B-sympcomp-identity}})\\
        &=  \left(\left(\weyl(\ket{\psi_i}) \cap \Zs\right) + \calX\right) \cap \calZ && (\mathrm{\cref{lem:A-plus-B-sympcomp-identity}})\\
        &= (\calX \cap \calZ) + \left(\weyl(\ket \psi_i) \cap \Zs\right) && (\mathrm{\cref{lem:A-plus-B-sympcomp-identity-2}})\\
        &= \weyl(\ket{\psi_i}) \cap \Zs\\
        &= C_i(S_i).
    \end{align*}
    By the end of the $\mclifford$ random Clifford circuits, it follows that
    \[
    \hat{A} = \sum_{i=1}^{\mclifford} C_i^\dagger \left((\hat{H}_i + \Zs)^\sympcomp \right) \supseteq \sum_{i=1}^{\mclifford} S_i = \weyl(\ket \psi).
    \]
    
    Next, we argue that $\hat{A}$ is isotropic.  \cref{lem:sampling} shows that if we draw
    \[
    8\left(n + \log(3 \mclifford/\delta)\right) + 1
    \]
    computational difference samples per $\ket{\psi_i}$ then
    \begin{align}\label{eq:inner-loop-condition-2}
        \sum_{x \in \hat{H}_i} r_{\psi_i} > \frac{3}{4}
    \end{align}
    for all $\hat{H}_i$ with failure probability at most $\delta/3$.
    Again, assume that we succeed.
    By \cref{cor:r-mass-isotropic-2}, it follows that each $C_i^\dagger\left(\left(\hat{H}_i + \calZ\right)^\sympcomp\right) \subseteq \langle M \rangle$ where $M \coloneqq \{x \in \F_2^{2n} : 2^n p_\psi(x) > \frac{1}{2}\}$.
    The addition of these subspaces $\sum_i C^\dagger_i\left(\left(\hat{H}_i + \calZ\right)^\sympcomp\right)$ then must also lie in $\langle M \rangle$. By \cref{fact:M_half_commute}, $\langle M \rangle$ is isotropic, and so too must be $\hat{A}$.

    In stage two, we need to prove that \Cref{item:dimension,item:approx-weyl} are met by the output $S$.
    Since we assume that stage one succeeded, $\hat{A}$ is isotropic.
    Therefore, we can apply \cref{lem:clifford-mapping-algorithm} to find a Clifford circuit such that $C_\Zs\left(\hat{A}\right) \subseteq \calZ$, as is done in \cref{step:seven-2}.
    Since $\weyl(\ket \psi) \subseteq \hat{A}$, it also follows that $\weyl(\ket{\psi^\prime}) \subseteq \calZ$.

    We now start the second round of computational difference sampling in \cref{step:eight-2,step:nine-2}.
    By \cref{lem:sampling}, $\hat{H} \subseteq \calX$ from \cref{step:ten-2} will be such that \[\sum_{x \in \hat{H}} r_{\psi^\prime}(x) \geq 1-\frac{\eps^2}{4}\]
    with probability at least $\delta/3$. Once again, assume that we succeed. It follows that
    \begin{align*}
        \sum_{x \in S}p_\psi(x) &= 
        \sum_{x \in \left(\hat{H}+\calZ\right)^\sympcomp}p_{\psi'}(x) && (\text{\cref{fact:clifford-permutes-p-mass}})\\
        &= \frac{2^n}{|\hat{H} + \Zs|}\sum_{x \in \hat{H}} r_{\psi'}(x) && (\text{\cref{cor:r-mass-vs-p-mass-duality}})\\
        &\ge \left(1-\frac{\eps^2}{4}\right) \frac{2^n}{|\hat{H} + \Zs|}\\
        &= \left(1-\frac{\eps^2}{4}\right) \frac{|S|}{2^n}, && (\text{\cref{fact:sympcomp}})
    \end{align*}
    thus satisfying \Cref{item:approx-weyl}.

    Finally, \cref{cor:comp-diff-sampling} guarantees that we only sample from \[H \coloneqq \left(\weyl(\ket{\psi^\prime}) \cap \calZ\right)^\sympcomp \cap \calX = \weyl(\ket{\psi^\prime})^\sympcomp \cap \calX =  \left(\weyl(\ket{\psi^\prime}) + \calX\right)^\sympcomp.\]
    Hence, $\hat{H} \subseteq H$, and therefore
    \begin{align*}
    \left(\hat{H} + \calZ\right)^\sympcomp &\supseteq \left(H + \calZ\right)^\sympcomp\\
    &= H^\sympcomp \cap \calZ && (\mathrm{\cref{lem:A-plus-B-sympcomp-identity}})\\
    &= \left(\weyl(\ket{\psi^\prime}) + \calX\right) \cap \calZ\\
    &= \left(\calX \cap \calZ\right) + \weyl(\ket{\psi^\prime}) && (\mathrm{\cref{lem:A-plus-B-sympcomp-identity-2}})\\
    &= \weyl(\ket{\psi^\prime}).
    \end{align*}
    It follows that $S \supseteq \weyl(\ket \psi)$, and since $\dim\left(\weyl(\ket \psi)\right) \geq n-t$, we conclude that $S$ satisfies \Cref{item:dimension}.

    By the union bound, this means we satisfy \Cref{item:dimension,item:approx-weyl} with probability at least $1-\delta$.

    We next turn to the sample complexity analysis.
    For simplicity, let $\mcomp$ be the number of computational difference samples taken in each part of the inner loop.
    Because $\log\left(\mclifford\right) = O\left(t + \log n + \frac{\log(3/\delta)}{n}\right)$, we find that
    \[\mcomp \coloneqq 8\left(n + \log(3\mclifford/\delta)\right) + 1
    = O\left(n + \log(1/\delta)\right),
    \]
    which will greatly simplify our asymptotic analysis.
    Since a computational difference sample involves $2$ copies of $\ket \psi$, the total number of copies of $\ket \psi$ needed in stage 1 is
    \[
        2 \mclifford\mcomp= O\left(\left(n + \log(1/\delta)\right)^2 2^t\right).
    \]
    In stage 2 we use $\frac{16(n + \log(3/\delta))}{\eps^2}$ samples, so the total number of samples is \[ O\left(\left(n + \log(1/\delta)\right)^2 2^t + \frac{n + \log(1/\delta)}{\eps^2}\right).\]

    To analyze the runtime, we will first analyze the inner loop of \cref{alg:learning_weyl_single_adaptive}.
    By \cref{prop:random-clifford-sample}, we require $O(n^2)$ time to sample a random Clifford circuit in the first part of \cref{step:one-2}.
    Applying $C_i$ to $2 \mcomp$ copies of $\ket \psi$ in the second half \cref{step:one-2} requires $O(\mcomp n^2)$ time.
    In \cref{step:two-2}, we require $O(\mcomp n)$ time to make $O(\mcomp)$ many $n$-qubit measurements, then apply the difference to pairs of measurements.
    Finally, it costs $O(\mcomp n^2)$ time to run Gaussian elimination in \cref{step:three-2} to find the generators of each $\hat{H}_i$.
    Thus, the inner loop requires time $O(\mcomp n^2)$, which becomes
    \[
        O(\mclifford \mcomp n^2) = O\left(n^2\left(n + \log(1/\delta)\right)^2 2^t\right)
    \]
    over all iterations of the inner loop.
    This will be one of the dominating terms.
    
    Analyzing from \cref{step:six-2} and onwards, it takes takes $O(\mclifford n^3)$ time to apply \cref{lemma:compute-symplectic-complement} to each of the $\hat{H}_i + \calZ$.
    It similarly  requires $O(\mclifford n^3)$ time to apply $C_i^\dagger$ to each $\left(\hat{H}_i + \calZ\right)^\sympcomp$.
    \Cref{lem:clifford-mapping-algorithm} requires $O(n^2)$ time in \cref{step:seven-2} to produce a circuit of size at most $O(n^2)$.
    As a result, the computational difference sampling in \cref{step:eight-2,step:nine-2} requires $O\left(\frac{n^2(n + \log(1/\delta)}{\eps^2}\right)$ time.
    Finally, it takes $O\left(\frac{n^2(n+\log(1/\delta))}{\eps^2}\right)$ time to perform \cref{lemma:compute-symplectic-complement} one last time in \cref{step:ten-2}. This is the second dominating term. \qedhere
\end{proof}

\section{Generalization to Mixed States}\label{appendix:mixed-states}

We show that \cref{alg:reduction,alg:learning_weyl_bell} remain correct when run on mixed states. 
For a mixed state $\rho$, we define $\weyl(\rho) \coloneqq \{x \in \F_2^{2n} : \tr(W_x \rho) = \pm 1 \}$.
It is easy to see that $\weyl(\rho)$ is still an isotropic subspace, just as it is for pure states. 
The stabilizer dimension of $\rho$ is then the dimension of the subspace $\weyl(\rho)$.
We will prove the following theorem.

\begin{theorem}\label{thm:mixed-tomography}
Given $n, t \in \mathbb{N}$, there is an algorithm that uses $\poly(n, 2^t, 1/\eps)$ time and $\poly(n, 1/\eps)$ copies of an $n$-qubit mixed state $\rho$ and learns $\rho$ to trace distance $\eps$ with high probability, promised that $\rho$ has stabilizer dimension at least $n - t$. 
\end{theorem}

We must recall the trace distance and fidelity between mixed quantum states. 
The trace distance between two mixed states $\rho, \sigma$ is $d_\tr(\rho, \sigma) \coloneqq \frac{1}{2}\norm{\rho - \sigma}_1$, where, for a matrix $A$, $\lVert A \rVert_1$ is the sum of the absolute values of the singular values of $A$.
The fidelity between $\rho$ and $\sigma$ is $\fidelity(\rho, \sigma) \coloneqq \tr\left(\sqrt{\sqrt{\rho} \sigma \sqrt{\rho}}\right)^2$.

\subsection{Bell Difference Sampling Mixed States}
To prove \cref{thm:mixed-tomography}, we must develop the theory of Bell difference sampling for mixed states. 
As we will show, many properties that hold for pure states remain true for mixed states. 
Indeed, we believe many of the results/algorithms in the literature that are based on Bell difference sampling can be generalized to mixed states if one desires to.

Recall from \cite[Equation 3.7]{gross2021schur} that one can describe Bell difference sampling as a projective measurement. 
In particular, upon Bell difference sampling a mixed state $\rho$, one observes $a \in \F_2^{2n}$ with probability $\tr(\Pi_a \rho^{\otimes 4})$, where $\Pi_a$ is defined as
\begin{align}\label{eq:bell-diff}
\Pi_a \coloneqq \frac{1}{4^n}\sum_{x \in \F_2^{2n}} (-1)^{[a, x]}W_x^{\otimes 4}.
\end{align}
Hence, we define $q_\rho(x) \coloneqq \tr(\Pi_{x}\rho^{\otimes 4})$, the distribution induced by the Bell difference sampling measurement. 

We find it convenient to define $p_\rho(x) \coloneqq 2^{-n} \tr(W_x \rho)^2$ in analogy with $p_\psi$. 
However, many properties of $p_\psi$ fail to carry over to $p_\rho$. For example, $p_\rho$ is no longer a probability distribution, so one can no longer express $q_\rho$ as the convolution of $p_\rho$ with itself as a result.
Note that for any $x \in \F_2^{2n}$, we still have that $0 \leq p_\rho(x) \leq 2^{-n}$. 

A crucial property of $q_\psi$ is that it exhibits a duality property (\cref{thm:p_q_duality}). 
We prove that this property holds for $q_\rho$ as well.

\begin{theorem}[Generalization of \Cref{thm:p_q_duality}]\label{thm:duality-q-mixed}
Let $A \subseteq \F_2^{2n}$ be a subspace. Then
    \[
        \sum_{a \in A} q_\rho(a) = \abs{A} \sum_{x \in A^\sympcomp} p_\rho(a)^2.
    \]
\end{theorem}
\begin{proof}
We have 
\begin{align*}
\sum_{a \in A} q_\rho(a) 
&= \sum_{a \in A} \tr(\Pi_a \rho^{\otimes 4}) \\    
&= \sum_{a \in A} \sum_{x \in \F_2^{2n}} (-1)^{[a,x]} p_\rho(x)^2 \\    
&= \sum_{x \in \F_2^{2n}} p_\rho(x)^2 \sum_{a \in A} (-1)^{[a,x]} \\    
&= \abs{A} \sum_{x \in \F_2^{2n}} p_\rho(x)^2 \indic{x \in A^\sympcomp}  \\    
&= \abs{A} \sum_{x \in A^\sympcomp} p_\rho(x)^2. 
\end{align*}
Above, the first three lines follow from \cref{eq:bell-diff}, the linearity of the trace, and some basic arithmetic manipulations. 
The fourth line follows from the fact that for any subspace $T \subseteq \F_2^{2n}$ and a fixed $x \in \F_2^{2n}$,
    \[\sum_{a \in T} (-1)^{[a, x]} = \abs{T} \cdot \indic{x \in T^\perp}.\]
See, for example, \cite[Lemma 2.11]{grewal2023improved} for a proof of this fact.  
Readers familiar with Boolean Fourier analysis will recognize this as a basic fact about summing over (symplectic) Fourier characters. 
\end{proof}

The duality theorem for $q_\rho$ allows us to prove the following two statements. 
The first states that the support of $q_\rho$ is contained in $\weyl(\rho)^\sympcomp$, generalizing \cite[Corollary 4.4]{grewal2023improved}.
The second bounds the $q_\rho$-mass on a subspace $A$ in terms of $p_\rho$.

\begin{lemma}\label{lemma:support-weyl-perp-mixed}
For any quantum mixed state $\rho$, the support of $q_\rho$ is contained in $\weyl(\rho)^\sympcomp$.
\end{lemma}
\begin{proof}
    \begin{align*}
        \sum_{x \in \weyl(\rho)^\sympcomp} q_\rho(x) &= \abs{\weyl(\rho)^\sympcomp} \sum_{x \in \weyl(\rho)} p_\rho(x)^2 && (\text{\cref{thm:duality-q-mixed}})\\
        &= \abs{\weyl(\rho)^\sympcomp} \cdot \frac{\abs{\weyl(\rho)}}{4^n}\\ 
        &= 1.
    \end{align*}
    The second line follows from that fact that for $x \in \weyl(\rho)$, $p_\rho(x)^2 = 4^{-n}$.
\end{proof}

\begin{proposition}
\label{prop:q-lower-bounds-p-mixed}
Let $A \subseteq \F_2^{2n}$ be a subspace. Then
\[\sum_{a \in A} q_\rho(a) \leq \frac{\abs{A}}{2^n}\sum_{a \in A^\sympcomp} p_\rho(a).\]
\end{proposition}
\begin{proof}
    \begin{align*}
        \sum_{a \in A} q_\rho(a) &= \abs{A} \sum_{x \in A^\perp} p_\rho(x)^2 && \text{(\cref{thm:duality-q-mixed})}\\
        & \leq \frac{\abs{A}}{2^n} \sum_{x \in A^\perp} p_\rho(x). && (p_\rho(x) \leq \frac{1}{2^n}) && \qedhere
    \end{align*}
\end{proof}

\subsection{Reduction to Finding a Heavy Subspace}
We argue in \cref{thm:learn-weyl-psi-reduction} that efficient tomography of large-stabilizer-dimension states reduces to finding a subspace of $\F_2^{2n}$ that isn't too large and has sufficiently large $p_\psi$ mass. 
In this section, we prove that this reduction holds for mixed states as well.
The formal statement of this follows.

\begin{theorem}[Generalization of \cref{thm:learn-weyl-psi-reduction}]
\label{thm:learn-weyl-psi-reduction-mixed}
Let $\rho$ be an $n$-qubit mixed state, and suppose there exists a subspace $S \subseteq \F_2^{2n}$ that satisfies the following conditions: 
\begin{enumerate}[label=\rm{(\roman*)}]
    \item\label{item:dimension-mixed} $\dim(S) \geq n-t$, and
    \item\label{item:approx-weyl-mixed} $\sum_{x \in S} p_\rho(x) \geq \left(1-\frac{\eps^2}{4}\right)\frac{\abs{S}}{2^n}$.
\end{enumerate}
Then, given $S$ and copies of $\rho$ as input, \cref{alg:reduction} outputs a classical description of $\hat{\rho}$ such that $\tracedistance{\rho, \hat{\rho}} \leq \eps$ with probability at least $1-\delta$ using $\frac{4}{3}N_{t, \frac{\eps}{2}, \frac{\delta}{3}} + \frac{224}{9} \log(3/\delta)$ copies of $\rho$, and 
    $M_{t, \frac{\eps}{2}, \frac{\delta}{3}} + O\left(n^2 \left(N_{t, \frac{\eps}{2}, \frac{\delta}{3}} + \log(1/\delta)\right)\right)$  time, while performing only single-copy measurements.
\end{theorem}

\begin{remark}\label{remark:mixed-state-tomography-constants}
In this appendix, $M_{n,\eps,\delta}$ and $N_{n,\eps,\delta}$ above are the same as in \cref{def:tomography-complexity} except for mixed state tomography. That is, we let $M_{n,\eps,\delta}$ and $N_{n,\eps,\delta}$ denote the copy and time complexities, respectively, of state tomography of an $n$-qubit mixed state to trace distance $\eps$ and success probability at least $1-\delta$. 
As for pure states, one can use the state tomography algorithms in \cite{franca_et_al:LIPIcs.TQC.2021.7} or \cite[Section 5]{aaronson2023efficient} to instantiate \cref{thm:learn-weyl-psi-reduction-mixed}.
\end{remark}

\begin{proof}
The only aspects of the proof of \cref{thm:learn-weyl-psi-reduction} that use the purity of the input are the \cref{lem:p-mass-isotropic,lem:product-state-approximation-general}, and the identity $\tracedistance{\ket \psi, \ket \phi} = \sqrt{1-\fidelity(\ket{\psi},\ket{ \phi})}$. 
We have generalized the first two statements below to mixed states in \cref{lem:p-mass-isotropic-mixed,lem:product-state-approximation-general-mixed}, and, for two mixed states $\rho, \sigma$, 
$\tracedistance{\rho,\sigma} \leq \sqrt{1-\fidelity(\rho,\sigma)}$ \cite{fuchs1999cryptographic,Watrous_2018}. 
Hence, the proof goes through with no loss in parameters.
\end{proof}

To complete the proof of \cref{thm:learn-weyl-psi-reduction-mixed}, it remains to generalize \cref{lem:p-mass-isotropic,lem:product-state-approximation-general} to mixed states.
We prove these generalizations in the remainder of this section.

\begin{lemma}[Generalization of \cref{lem:p-mass-isotropic}]
\label{lem:p-mass-isotropic-mixed}
    Let $S$ be a subspace of $\F_2^{2n}$ such that
    \[
    \sum_{x \in S}p_{\rho}(x) > \frac{3}{4}\frac{\abs{S}}{2^n}.
    \]
    Then $S$ is isotropic.
\end{lemma}
\begin{proof}
    Let $A \coloneqq \{x \in S: 2^n p_\rho(x) > 1/2\}$.
    Because \cref{fact:uncertainty} holds for mixed states, it follows that \cref{fact:M_half_commute} holds as well.  
    Therefore, the rest of the argument in \cref{lem:p-mass-isotropic} follows, because, aside from the above statements, the proof only uses facts about symplectic vector spaces.
\end{proof}

The final two lemmas of this subsection generalize \cref{lem:product-state-approximation-general} to mixed states.

\begin{lemma}[Generalization of \cref{lem:product-state-approximation}]
\label{lem:product-state-approximation-mixed}
    Let $S = 0^{n+t} \times \F_2^{n-t}$, and suppose that 
    \[\sum_{x \in S} p_\rho(x) \geq \frac{1 - \eps}{2^t}.\] 
    Then there exists an $(n-t)$-qubit computational basis state $\ket{x}$ and a $t$-qubit quantum state 
    \begin{align} \label{eq:sigma-form}
     \sigma \coloneqq \frac{\trc_B\left((I \otimes \ketbra{x}{x}) \rho \right)}{\trc\left((I \otimes \ketbra{x}{x}) \rho \right)},
    \end{align}
    such that
    the fidelity between $\sigma \otimes \ketbra{x}{x}$ and $\rho$ is at least $1 - \eps$. The partial trace in \cref{eq:sigma-form} is over the last $n-t$ qubits.\footnote{This is to say that $\sigma$ is obtained by postselecting on measuring the last $n - t$ qubits of $\rho$ to be $\ket{x}$.}
\end{lemma}
\begin{proof}
    We can always write: 
    \[
    \rho = \sum_{y,z \in \F_2^{n-t}} \alpha_{y,z} \cdot \sigma_{y,z} \otimes \ketbra{y}{z}.
    \]
    Here, the $\alpha_{z,z}$ are the diagonal elements of $\tr_A(\rho)$, where $A$ corresponds to the first $t$ qubits, so they are nonnegative real numbers that sum to $1$. 
    The $\sigma_{y,z}$ are $t$-qubit mixed states.
    
    Define $x \coloneqq \argmax_{z \in \F_2^{n-t}} \alpha_{z,z}$.
    Observe that $\tr_B((I \otimes \ketbra{x}{x}) \rho) = \alpha_{x,x} \cdot \sigma_{x,x}$.
    We will argue that the fidelity between between $\sigma_{x,x} \otimes \ketbra{x}{x}$ and $\rho$ is equal to $\alpha_{x,x}$. Then we will argue that $\alpha_{x,x} \geq 1 - \eps$. Hence, taking $\sigma \coloneqq \sigma_{x,x}$ gives the desired state claimed in \cref{eq:sigma-form}.
    
    First, we prove that  $\fidelity(\rho,\sigma_{x,x} \otimes \ketbra{x}{x})=\alpha_{x,x}$. 
    We have 
    \begin{align*}
        \fidelity\left(\rho, \sigma_{x,x} \otimes \ketbra{x}{x}\right) 
        &= \tr\left(\sqrt{ \sqrt{\sigma_{x, x}} \otimes \ketbra{x}{x} \cdot \left(\sum_{y,z \in \F_2^{n-t}} \alpha_{y,z} \cdot \sigma_{y,z} \otimes \ketbra{y}{z}\right)  \cdot \sqrt{\sigma_{x, x}} \otimes \ketbra{x}{x}}\right)^2 \\
        &= \tr\left(\sqrt{ \sum_{y, z \in \F_2^{n-t}} \alpha_{y,z}  \cdot \left(\sqrt{\sigma_{x,x}} \sigma_{y,z}\sqrt{\sigma_{x,x}}\right) \otimes \left(\ket{x}\!\braket{x|y}\!\braket{z|x}\!\bra{x}\right)}\right)^2 \\
        &= \tr\left(\sqrt{\alpha_{x,x}  \cdot \sigma_{x,x}^2 \otimes \ketbra{x}{x}}\right)^2 \\
        &= \alpha_{x,x} \cdot \tr\left(\sigma_{x,x} \otimes \ketbra{x}{x}\right)^2 \\ 
        &= \alpha_{x,x}.
    \end{align*}
Next, we prove that $\alpha_{x,x} \geq 1 - \eps$.
    \begin{align*}
        1 - \eps &\leq  2^t \sum_{z \in S} p_\rho(x) \\
        &= \frac{1}{2^{n-t}} \sum_{z \in \F_2^{n-t}} \tr\left(I \otimes W_z \otimes I \otimes W_z \cdot \rho^{\otimes 2}\right) \\
        &= \sum_{z \in \F_2^{n-t}} \tr\left(\left(I \otimes \ketbra{z}{z}\right) \rho \right)^2 && (\text{\cref{prop:sum-over-stabilizer-basis}}) \\
        &\leq \left(\max_{z \in \F_2^{n-t}} \tr\left(\left(I \otimes \ketbra{z}{z}\right) \rho \right)\right) \cdot \sum_{z \in \F_2^{n-t}}\tr\left(\left(I \otimes \ketbra{z}{z}\right) \rho \right) \\
        &= \alpha_{x, x} \cdot \sum_{z \in \F_2^{n-t}} \alpha_{z, z} \\
        &= \alpha_{x,x}. 
    \end{align*}
On the fourth line, we used 
the fact that $\tr((I \otimes \ketbra{z}{z}) \rho) = \alpha_{z,z}$. 
\end{proof}

 \begin{lemma}[Generalization of \cref{lem:product-state-approximation-general}]
        \label{lem:product-state-approximation-general-mixed}
        Let $S$ be an isotropic subspace of dimension $n-t$, and suppose that 
    \[\sum_{x \in S} p_\rho(x) \geq \frac{1 - \eps}{2^t}.\] 
    Then there exists a state $\hat{\rho}$ with $S \subseteq \weyl(\hat{\rho})$ such that the fidelity between $\hat{\rho}$ and $\rho$ is at least $1 - \eps$. 
    
    In particular, $\hat{\rho} = C^\dagger (\sigma \otimes \ketbra{x}{x})C$, where $\ket{x}$ is an $(n-t)$-qubit basis state, 
    \begin{align}\label{eq:post-selected-mixed}
     \sigma = \frac{\trc_B\left((I \otimes \ketbra{x}{x}) C\rho C^\dagger \right)}{\trc\left((I \otimes \ketbra{x}{x}) C\rho C^\dagger \right)}
    \end{align}
    is a $t$-qubit quantum state, $C$ is the Clifford circuit mapping $S$ to $0^{n+t} \times \F_2^{n-t}$ described in \cref{lem:clifford-mapping-algorithm}, and the partial trace in \cref{eq:post-selected-mixed} is over the last $n-t$ qubits.
    \end{lemma}
    \begin{proof}
    Following the proof of \cref{lem:product-state-approximation-general}, the Clifford circuit that maps $S$ to $0^{n+t} \times \F_2^{n-t}$ permutes the $p_\rho$ mass in the same way that it permutes the $p_\psi$ mass for a pure state as per \cref{fact:clifford-permutes-p-mass}.
    That is, for $\rho^\prime \coloneqq C \rho C^\dagger$, \[p_{\rho^{\prime}}(x) = \frac{1}{2^n} \tr\left(W_x C \rho C^\dagger\right) = p_\rho(C^\dagger(x)).\]
    The only remaining step of the proof that relies on a pure state is when \cref{lem:product-state-approximation} is used. 
    Since we generalized \cref{lem:product-state-approximation} to hold for mixed states, via \cref{lem:product-state-approximation-mixed}, the proof goes through.
    \end{proof}

\subsection{Efficient Tomography of Large-Stabilizer-Dimension Mixed States}

We give the remaining details of our efficient tomography algorithm. 
The final missing piece is to show that one can learn a subspace satisfying the two conditions in \cref{thm:learn-weyl-psi-reduction-mixed} by repeatedly Bell difference sampling copies of a mixed state.

\begin{theorem}[Generalization of \cref{thm:bell-copy-main}]\label{thm:bell-copy-main-mixed}
    Let $\rho$ be an $n$-qubit quantum state with stabilizer dimension at least $n-t$.
    Given copies of $\rho$, \cref{alg:learning_weyl_bell} succeeds at outputting a subspace $S$ that satisfies the conditions of \cref{thm:learn-weyl-psi-reduction-mixed} with probability at least $1-\delta$ with no loss in parameters.
    In particular,
    the algorithm uses
    \[
    \frac{32\log(1/\delta)+64n}{\eps^2}
    \]
    samples and
    \[
    O\left(\frac{n^2\log(1/\delta) + n^3}{\eps^2}\right),
    \]
    time.
\end{theorem}

\begin{proof}
    By \Cref{lemma:support-weyl-perp-mixed}, the support of $q_\rho$ is a subspace of dimension at most $n + t$, so $\dim S^\sympcomp \le n + t$ and therefore $\dim S \ge n - t$, by \cref{fact:sympcomp}. This establishes \cref{item:dimension-mixed}.

   By \Cref{lem:sampling}, except with probability at most $\delta$, the Bell difference sampling phase of the algorithm in \Cref{step:one_bell} finds a subspace $S^\sympcomp \subseteq \F_2^{2n}$ such that 
    \[
    \sum_{x \in S^\sympcomp} q_\rho(x) \geq 1-\frac{\eps^2}{4}.
    \]
   From \cref{prop:q-lower-bounds-p-mixed}, we know that 
    \[
   \frac{\abs{S^\perp}}{2^n}  \sum_{x \in S} p_\rho(x) \geq 1-\frac{\eps^2}{4}, 
    \]
    and because $\abs{S}\cdot\abs{S^\perp} = 4^n$, we have 
    \[
   \sum_{x \in S} p_\rho(x) \geq \left(1-\frac{\eps^2}{4}\right) \frac{\abs{S}}{2^n}, 
    \]
    which establishes \cref{item:approx-weyl-mixed}.

    The running time analysis is the same as the analysis given in the proof of \cref{thm:bell-copy-main}.
\end{proof}

We can now combine \cref{thm:bell-copy-main-mixed,thm:learn-weyl-psi-reduction-mixed} to prove the main theorem of this appendix.
We note that the generalization to mixed states goes through with no loss in parameters (except potentially for the parameters referenced in \cref{remark:mixed-state-tomography-constants}).

\begin{corollary}[Formal version of \cref{thm:mixed-tomography}]\label{cor:bell-copy-main-mixed}
    Let $\rho$ be an $n$-qubit quantum state with stabilizer dimension at least $n-t$.
    Given copies of $\rho$ as input, the combination of \cref{alg:learning_weyl_bell,alg:reduction} outputs a classical description of $\hat{\rho}$ such that $\tracedistance{\rho, \hat{\rho}} \leq \eps$ with probability at least $1-\delta$.
    The algorithm uses
    \[
    \frac{32\log(2/\delta)+64n}{\eps^2} + \frac{4}{3}N_{t, \frac{\eps}{2}, \frac{\delta}{6}} + \frac{224}{9}\log(6/\delta)
    \]
    samples and
    \[
    O\left(\frac{n^2\log(1/\delta) + n^3}{\eps^2} + n^2 \cdot N_{t, \frac{\eps}{2}, \frac{\delta}{6}} \right) + M_{t, \frac{\eps}{2}, \frac{\delta}{6}},
    \]
    time.
\end{corollary}
\begin{proof}
    By \cref{thm:learn-weyl-psi-reduction-mixed,thm:bell-copy-main-mixed}, each algorithm succeeds with probability at least $1-\delta/2$. By the union bound, the total failure probability is at most $\delta$.
\end{proof}

As a final remark, we note that \cref{cor:bell-copy-main-mixed} learns the input state in \emph{fidelity}, which, for mixed states, is a stronger result than trace distance tomography. 
This is apparent in the Fuchs-van de Graaf inequality $d_\tr(\rho,\sigma) \leq \sqrt{1- \fidelity(\rho,\sigma)}$, which says that if $\rho$ and $\sigma$ are close in fidelity, then they are also close in trace distance. 
Hence, if one chooses a fidelity state tomography algorithm in \cref{thm:learn-weyl-psi-reduction-mixed}, the final output of \cref{cor:bell-copy-main-mixed} will be close to the input in fidelity.

\end{document}